\documentclass[a4paper, 10pt, twosides]{amsart}
\usepackage{setspace}
\usepackage{moreverb}
\usepackage[colorlinks=true,
        raiselinks=true,
        linkcolor=red,
        citecolor=red,
        urlcolor=red,
        plainpages=false]{hyperref}

\usepackage{amsmath}
\usepackage{amssymb}
\usepackage{amsfonts}
\usepackage{amsthm}
\usepackage{mathtools}
\usepackage{thmtools, thm-restate}
\usepackage{enumitem}
\usepackage{dsfont}
\usepackage{color}
\makeatletter
\def\BState{\State\hskip-\ALG@thistlm}
\makeatother
\usepackage[retainorgcmds]{IEEEtrantools}
\allowdisplaybreaks
\usepackage{graphics}
\usepackage{multirow}
\usepackage{empheq}
\usepackage{verbatim}
\usepackage[latin1]{inputenc}
\usepackage{tikz}
\usetikzlibrary{shapes,arrows}
\usepackage{float}
\usepackage[squaren]{SIunits}
\usepackage{subfigure}
\usepackage{algorithm}
\usepackage{algpseudocode}
\usepackage{pifont}
\theoremstyle{definition}

\theoremstyle{remark}

\numberwithin{equation}{section}

\newcommand{\I}{I_{3\times 3}} % identity matrix
\DeclareMathOperator*{\minimize}{minimize}
\DeclareMathOperator*{\maximize}{maximize}

\DeclareMathOperator{\trace}{tr}
\DeclareMathOperator{\sgn}{sgn}
\DeclareMathOperator{\e}{e}
\DeclareMathOperator{\logm}{logm}
\newcommand{\Orntcnst}{\mathcal{C}_{\text{ornt}}}
\newcommand{\Mtmcnst}{\mathcal{C}_{\text{mtm}}}
\newcommand{\cltb}{c} % control bound
\newcommand{\mtmb}{d} % momentum bound
\newcommand{\inertia}{J} % inertia of the body
\newcommand{\mint}{J_d} % modified inertia matrix
\newcommand{\R}{\mathds{R}} %real numbers
\newcommand{\defas}{:=}
\newcommand{\nsp}{\kern-1em}
\newcommand{\iso}{\ \raisebox{0.25ex}{\ensuremath{\stackrel{\sim}{\rule{8pt}{0.35pt}}}\ \,}}
\newcommand{\norm}[1]{\left\lVert{#1}\right\rVert} %norm
\newcommand{\abs}[1]{\left\lvert{#1}\right\rvert} %absolute value
\newcommand{\so}[1]{\mathfrak{so}(#1)} % small so(3)
\newcommand{\SO}[1]{\text{SO}(#1)} % SO(3) manifold
\newcommand{\secref}[1]{\S\ref{#1}} % section ref
\newcommand{\ip}[2]{\left\langle#1,#2\right\rangle} % Inner product
\newcommand{\mtm}[1]{\Pi_{#1}} % Momentum variable
\newcommand{\pmtm}[1]{\tilde{\Pi}_{#1}} % Projected Momentum variable
\newcommand{\pmcom}[1]{\tilde{\gamma}_{#1}} % Modified Costate corresponding to momentum
\newcommand{\rot}[1]{F_{#1}} % small rotation F_k
\newcommand{\ornt}[1]{R_{#1}} % Rotation matrix representing orientation
\newcommand{\m}[1]{\( #1 \)}  % mathmode
\newcommand{\cltr}[1]{u_{#1}} % control
\newcommand{\ocltr}[1]{\mathring{u}_{#1}} % optimal control
\newcommand{\rvec}[1]{\begin{pmatrix} #1^1 & #1^2 & #1^3\end{pmatrix}} % Row vector
\newcommand{\com}[1]{\lambda_{#1}} % Costate corresponding to momentum
\newcommand{\cor}[1]{\chi_{#1}} % CoState corresponding to rotation
\newcommand{\mcom}[1]{\gamma_{#1}} % Modified Costate corresponding to momentum
\newcommand{\mcor}[1]{\zeta_{#1}}   % Modified CoState corresponding to rotation
\newcommand{\clts}[1]{\alpha_{#1}} % slack variable corresponding to control
\newcommand{\mtms}[1]{\beta_{#1}}  % slack variable corresponding to momentum
\newcommand{\tr}[1]{\left(\trace\left(\rot{#1}\right)\I-\rot{#1}\right)} % Trace part of expression
\newcommand{\mtr}[1]{\left(\trace\left(\mint \rot{#1}^\top\right)\I-\mint \rot{#1}^\top\right)} % trace expression with J_d
\newcommand{\mtrt}[1]{\left(\trace\left(\rot{#1}\mint\right)\I-\rot{#1}\mint\right)} % transpose of \mtr
\newcommand{\ii}[1]{\mathrm{I}^{#1}} % component of inertia matrix
\newcommand{\der}[1]{\mathcal{D}_{#1}} % Derivative of expression
\newcommand{\drot}[1]{\der{\mtm{#1}^i}\rot{#1}} % Derivative of \rot w.r.t. momentum
\begin{document}

%\runningheads{{\it FIRST AUTHOR}, et al.}{Optimal attitude control of a spacecraft with state constraints}
%\itshape{\journalnamelc \footnotemark[2]}  
\title[Attitude control with state and control constraints]{Discrete-time optimal attitude control of spacecraft with momentum and control constraints}
\author{Karmvir Singh Phogat, Debasish Chatterjee, Ravi N.~Banavar}
\address{Indian Institute of Technology Bombay, Mumbai, India.}
%\corraddr{Karmvir Singh Phogat, System and Control Engineering, IIT Bombay, Room 101, ACRE Building, IIT Powai, Mumbai~--~400 076, India. E-mail: karmvir.p@iitb.ac.in.}
%\cgs{The authors were partially supported by the Indian Space Research Organization grant 14ISROC010.}
%\cgs{14ISROC010 /Indian Space Research Organization.}
%%%%%%%%%%%%%%%%%%%%%%%%%%%%%%%%%%%%%%%%%%%%%%%%%%%%%%%%%%%%%%%%%%%%%%%%%%%%%%%%
\begin{abstract}
	This article solves an optimal control problem arising in attitude control of a spacecraft under state and control constraints. We first derive the discrete-time attitude dynamics by employing discrete mechanics. The orientation transfer, with initial and final values of the orientation and momentum and the time duration being specified, is posed as an energy optimal control problem in discrete-time subject to momentum and control constraints. Using variational analysis directly on the Lie group \m{\SO{3}\,\!,} we derive first order necessary conditions for optimality that leads to a constrained two point boundary value problem. This two point boundary value problem is solved via a novel multiple shooting technique that employs a root finding Newton algorithm. Robustness of the multiple shooting technique is demonstrated through a few representative numerical experiments. 
\end{abstract}

%%%%%%%%%%%%%%%%%%%%%%%%%%%%%%%%%%%%%%%%%%%%%%%%%%%%%%%%%%%%%%%%%%%%%%%%%%%%%%%%


\maketitle

%\footnotetext[2]{E-mail:}
%\href{http://www3.interscience.wiley.com/journal/5510/home}{\texttt{http://www3.interscience.wiley.com/journal/5510/home}}

\section{Introduction}\label{s:intro}
Typical space applications require reorienting a spacecraft or a satellite from a given initial configuration to a given final configuration over a given duration of time. Examples of such manoeuvres include positioning star sensors (attitude estimation sensors) towards deep space, pointing a camera in the desired direction for imaging purposes, positioning solar panels for effective tracking of the sun for optimal energy harvesting, etc~\cite{scrivener}. Such orientation maneuvers are popularly known as attitude maneuvers, and the class of attitude maneuvers that optimize a certain performance objective are termed optimal attitude maneuvers. In most applications, attitude control problems are subjected to additional constraints due to actuator saturation and/or  forbidden regions of the state space. Control constraints are omnipresent --- for instance, a momentum delivering device fitted on-board a spacecraft, e.g., a thruster, or a reaction wheel, has limitations in terms of maximum torque produced and maximum momentum delivered. In case a spacecraft is fitted with a flexible structure like solar panels, it must inevitably have a restriction of its maximum momentum for the protection of such equipments; such restrictions translate to state constraints in an optimal control problem. The class of constrained control problems involving orientation manoeuvres provides a rich supply of interesting and complex control problems related to space applications. Few computationally tractable solutions to such problems are known today, and in this article we propose a new and promising technique to solve a class of such constrained optimal control problems.
    
    In general, computationally tractable solutions to optimal control problems subject to nonlinear dynamics are difficult to arrive at \cite{trelat}. Constrained optimal attitude control problems are vastly more challenging than standard optimal control problems on Euclidean spaces because (a) the attitude kinematic equations evolve on the manifold \m{\SO{3}}, and classical techniques developed specifically for dynamics on Euclidean spaces, therefore, no longer directly apply, and (b) the presence of state and control constraints in the optimal control problem typically leads to two point boundary value problems (which represents first order necessary optimality conditions,) subject to inequality constraints, and these problems cannot be solved using conventional indirect techniques such as shooting methods \cite{trelat}. In this article we provide an algorithmic solution to one such class of constrained optimal attitude control problems subject to state and control constraints, where the attitude kinematics in discrete time are treated directly on the Lie group \m{\SO{3}} and are derived using discrete mechanics \cite{dm_marsden}.     
    
    Constrained optimal attitude control problems treated in the literature typically fall into one of two categories: one that considers only control constraints and the other that considers state constraints in addition to control constraints. Attitude control problems with control constraints have been discussed under the framework of optimal control with various performance indices involving minimization of time, fuel, and energy. A considerable body of work on the time-optimal attitude control problems exist~\cite{scrivener,ross, tsiotras, fleming2010}. Fuel or energy optimal maneuvers are also of great interest since on-board energy sources are limited and precious, and there is therefore a natural necessity  for using them in an optimal way~\cite{sagariss}. Early works on energy optimal control \cite{ athans,ilya} and fuel optimal control \cite{dixon,donghoon} tackled the continuous-time optimal attitude control problem, arriving at first order necessary conditions for optimality via the Pontryagin maximum principle. The resulting boundary value problems were typically solved using shooting or neighborhood extremal methods, both of which suffer from high sensitivity to initial data.       
	In recent years, attitude control problems with joint state and control constraints have been attacked from various directions. Attitude manoeuvres with joint state and control constraints, considering a linearized dynamics model is addressed in \cite{diehl, kolmanovskyfixed}. In these works the problem is recast as a nonlinear optimization problem, and solved using direct multiple shooting methods or sequential quadratic programming \cite{boyd}. Constrained attitude control problems with attitude kinematics represented in quaternion form  was explored in \cite{rossim,ulee} using pseudospectral methods, and in \cite{zhuang} using particle swarm optimization. Representation of the attitude kinematics in quaternion form has a serious disadvantage of non-uniqueness; consequently, boundary conditions defined for attitude kinematics do not have a unique representation in quaternions, which presents a problem during computations. Another commonly used attitude representation is Euler angles, which suffers from the defect of singularity \cite{schaub}. In order to avoid issues such as singularity and non-uniqueness, geometric techniques have been developed to solve constrained attitude control problems \cite{dm_marsden, bloch2009, leok3d}. One such approach to attitude control problems in continuous time with control constraints is discussed in \cite{saccon}. Projection operators were used there to find a search direction on the underlying Lie group, while employing Newton's method to solve the continuous time two point boundary value problem obtained via the Pontryagin maximum principle on manifolds \cite[p.\ 165]{agrachev}. A different geometric technique to attitude control problem with state and control constraints is addressed in \cite{rohit} where the authors attempted to handle state inequality constraints using penalty functions. As is well known, penalty functions do not \emph{enforce} the state inequality constraints, but add penalties to the cost functions if the constraints are violated. This constrained optimal control problem leads to a  two point boundary value problem that represents first order necessary optimality condition. The boundary value problem is then solved by employing indirect single shooting method. None of the these works discussed the issue of digital implementation, a key aspect of which is the process of discretization of the boundary value problems in time: Euler's step and its derivatives are insufficient since these do not account for nor respect the underlying manifold structure of \m{\SO{3}}.
   
 	In this article we tackle the issue of discretization up front by deriving a discrete time model via discrete mechanics~\cite{ leok3d, kobilarov}, and then employing variational analysis to arrive at first order necessary conditions for optimality. Discrete time models obtained via discrete mechanics are more accurate than other standard discretization schemes such as Euler's step because they preserve certain invariance properties like kinetic energy, momentum, etc, of the system, and the computations can be done  directly on the manifold \m{\SO{3}}, (because this discretization respects the manifold \m{\SO{3}},) thereby eliminating the problems associated with parametric representations of attitude. The presence of state inequality constraints in our attitude control problem makes it challenging because the resulting boundary value problems obtained using variational analysis are subject to inequality constraints on the states. Such constrained boundary value problems cannot in general be solved using classical multiple shooting methods precisely because of those inequality constraints. A non-classical multiple shooting technique has recently been proposed in~\cite{gerdts} for solving constrained boundary value problems arising in optimal control problems with state constraints. There the complementary slackness conditions arising due to inequality constraints on the states are represented in the form of equality constraints using the Fischer-Burmeister function~\cite[Equation (2.7)]{gerdts} --- a technique that can be highly inefficient in terms of computation because all the inequality constraints are considered at each iteration irrespective of them being active or not. In contrast, we propose a multiple shooting algorithm that deals with state inequality constraints in such a way that the dimension of the boundary value problem remains the same even when the state inequality constraints are active, making it more efficient in terms of time and memory complexity.   
    
	We reiterate that the thrust of our contribution is towards computational tractability of optimal attitude manoeuvres. The advantages of our technique are five fold: First, the discrete time model derived via discrete mechanics avoids the need for discretization at later stages, and eliminates issues associated with parametric representations of the attitude kinematic. Second, an indirect multiple shooting method, which provides more accurate solution to the problem unlike direct techniques, is employed to solve the optimal control problem. Integration steps of the discrete kinematic equations represented by rotation matrices are calculated at each discrete instant, which provide exact attitude trajectories for manoeuvres unlike other standard schemes \cite{trelat}. Third, the multiple shooting method that we employ here is more robust to initial guesses as compared to indirect single shooting methods that are comparable in terms of accuracy \cite{trelat}. Fourth, the proposed algorithm, being a multiple shooting method, can be implemented on a parallel architecture for fast computation. Fifth, the discrete time model obtained using discrete mechanics results in a boundary value problem that can be reduced to difference equations with momentum and co-momentum dynamics only. The dimension of this reduced difference equation model is half of the original one, which further contributes to savings in terms of time and memory. Moreover, scaling of the co-state variables such that the co-states are invariant to step length selected for a given manoeuvre, improves the radius of convergence of the multiple shooting algorithm. 
    
    This article unfolds as follows: In \secref{s:dis_mod} we employ discrete mechanics to derive a discrete time model of the attitude dynamics. In \secref{s:ocdtad} the energy optimal control problem is posed as a discrete time optimal control problem, and first order necessary conditions are obtained using variational analysis. Then we discuss scaling of the variables and reduction of the dynamics to momentum and co-momentum variables in \secref{s:smr}. \secref{s:msm} contains an introduction to multiple shooting methods, and provides a solution to the system of difference equations represented in momentum variables presented in \secref{s:smr}. \secref{s:nsr} provides the numerical experiments for large angle manoeuvres with momentum and control constraints. We conclude in \secref{s:cfd} with a brief discussion of future work. The proofs of our results are presented in a consolidated fashion in the Appendices.

%%%%%%%%%%%%%%%%%%%%%%%%%%%%%%%%%%%%%%%%%%%%%%%%%%%%%%%%%%%%%%%%%%%%%%%%%%%%%%%%
%%%%%%%%%%%%%%%%%%%%%%%%%%%%%%%%%%%%%%%%%%%%%%%%%%%%%%%%%%%%%%%%%%%%%%%%%%%%%%%%
%%%%%%%%%%%%%%%%%%%%%%%%%%%%%%%%%%%%%%%%%%%%%%%%%%%%%%%%%%%%%%%%%%%%%%%%%%%%%%%%
    \section{Derivation of Discrete time model via discrete mechanics}\label{s:dis_mod}
    This section contains the modeling of the attitude dynamics of the spacecraft using discrete mechanics \cite{dm_marsden}. In most optimal control problems involving mechanical systems, some sort of discretization is performed in order to employ numerical techniques. In the case of discrete mechanics, the variational description is directly discretized, and discrete time equations are obtained. This approach is advantageous in comparison to usual discretizations of continuous time models because it preserves certain invariants of the system such as momentum and energy.

    For ease of understanding, a quick introduction to discrete mechanics is given here, followed by discrete time modeling of the attitude dynamics of the spacecraft.
    \subsection{Introduction}\label{ss:dis_intro}
     Consider a mechanical system with the configuration space \m{Q} as a smooth manifold. Then the velocity vectors lie on the tangent bundle \m{TQ} of the manifold \m{Q} and the Lagrangian for the system can be defined as \m{L:TQ \rightarrow \R}~\cite{marsden}. In discrete mechanics, the velocity phase space \m{TQ} is replaced by \m{Q\times Q} which is locally isomorphic to \m{TQ}. Let us consider an integral curve \m{q(t)} in the configuration space such that \m{q(0)=q_0} and \m{q(h)=q_1}, where \m{h} represents the integration step. Then, the discrete Lagrangian \m{L_d:Q\times Q \rightarrow \R,} which is an approximation of the action integral along the integral curve segment between \m{q_0} and \m{q_1,} can be defined as~\cite{dm_ober} 
\begin{align}\label{eq:dis_lag}
L_d\left(q_0,q_1\right) \approx \int_{0}^{h}L\left(q(t),\dot{q}(t)\right) dt. 
\end{align}
Pick \m{h > 0}, (this quantity plays the role of step length) and consider a grid of the time domain \m{T=Nh} as \m{\left\{ t_k = kh | k=0,1, \ldots, N\right\}} and the corresponding discrete path space \m{\mathcal{P}_d(Q)\defas \left\{q_d :\{t_k\}_{k=0}^{N} \rightarrow Q \right\}}. The discrete trajectory \m{q_d \in \mathcal{P}_d(Q)} is such that \m{q_d(t_k)=q_k}. Now, defining the discrete action sum \m{\mathfrak{G}_d} as
\[ \mathfrak{G}_d(q_d) \defas  \sum_{k=0}^{N-1} L_d\left(q_k,q_{k+1}\right).\]
Assuming \m{q_d(0)=q_0} and \m{q_d(t_N)=q_N} fixed, define the variations \m{\delta(q_d)} as \m{\delta(q_d(t_k))=\delta(q_k)\in T_{q_k}Q} which vanishes at the end points \m{\left(\text{i.e.,\;} \delta(q_0)=\delta(q_N)=0\right)}. A discrete path \m{q_d} is a stationary point of the discrete action sum \m{\mathfrak{G}_d} if \m{\der{q_d} \mathfrak{G}_d(q_d)\delta(q_d) =0} for all \m{\delta(q_d)}~\cite[p.\ 21]{holm}.  
This is equivalent to saying that the points \m{\{q_k\}} of the path \m{q_d} satisfies the discrete Euler-Lagrange equations, i.e.,
\begin{align}\label{eq:lagrangian_eqn}
D_2 L_d\left(q_{k-1},q_k\right) + D_1 L_d\left(q_k,q_{k+1}\right)=0 \quad \text{\;for all\;} k=1,2,\ldots,N-1,
\end{align}
where \m{D_i} is the derivative of the function with respect to the \m{i}th argument. \\
\indent    Notice that \eqref{eq:lagrangian_eqn} involves \m{q_{k-1},q_k \text{\;and\;}q_{k+1}} at \m{k}th instant of time, which means that the difference equations obtained are of second order. To arrive at the discrete time model with first order difference equations, one needs to use the discrete time analogue of the Hamiltonian formulation.  
To this end, recall that the continuous time Legendre transform is a map \m{\mathbb{F}L} from the Lagrangian state space \m{TQ} to the Hamiltonian phase space \m{T^{*}Q}. Similarly, the discrete time Legendre transforms \m{\mathbb{F}^{+}L_d,\mathbb{F}^{-}L_d: Q\times Q \rightarrow T^{*}Q} ~\cite{dm_marsden} can be defined as
\begin{align*}
&\mathbb{F}^{+}L_d\left(q_k,q_{k+1}\right) \mapsto (q_{k+1},p_{k+1}) = \left(q_k, D_{2} L_d\left(q_k,q_{k+1}\right)\right),\\
&\mathbb{F}^{-}L_d\left(q_k,q_{k+1}\right) \mapsto (q_{k},p_{k}) = \left(q_k, -D_{1} L_d\left(q_k,q_{k+1}\right)\right),
\end{align*}
which are maps from the discrete Lagrangian state space \m{Q \times Q} to the discrete Hamiltonian phase space \m{T^* Q}. The map \m{\mathbb{F}^{+}L_d} is the forward discrete Legendre transform which relates \m{\left(q_k,q_{k+1}\right)} to \m{T^{*}_{q_{k+1}}Q}, and the map \m{\mathbb{F}^{-}L_d} is the backward discrete Legendre transform which relates \m{\left(q_k,q_{k+1}\right)} to \m{T^{*}_{q_{k}}Q}.
Let the discrete Lagrangian map \m{F_{L_d}:Q\times Q \rightarrow Q \times Q} be defined as \m{F_{L_d}\left(q_k,q_{k+1}\right) = \left(q_{k+1},q_{k+2}\right)}; it defines the evolution of the dynamics on the discrete state space. Then the corresponding discrete Hamiltonian map \m{\tilde{F}_{L_d}:T^{*}Q \rightarrow T^{*}Q} can be defined in the following equivalent ways ~\cite{dm_marsden}:
\begin{align*}
\tilde{F}_{L_d}\defas \mathbb{F}^{\pm}L_d \circ F_{L_d} \circ \left(\mathbb{F}^{\pm}L_d \right)^{-1} = \mathbb{F}^{+}L_d \circ \left(\mathbb{F}^{-}L_d \right)^{-1};
\end{align*}  
this is clear from the commuting diagram shown in Figure \ref{fig:comdiag}.
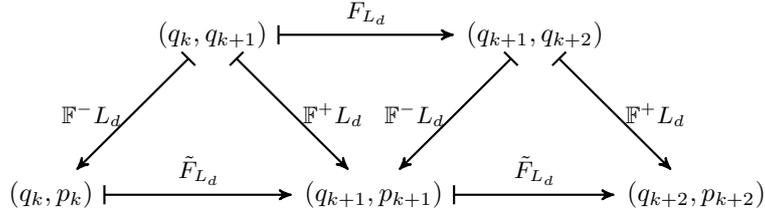
\begin{figure}[ht!]
\centering
\begin{tikzpicture}[|->,>=stealth',shorten >=1pt,auto,node distance=3cm,
  thick,main node/.style={font=\sffamily\bfseries}]
  \node[main node] (1) {$\left(q_k,q_{k+1}\right)$};
  \node[main node] (4) [below right of=1] {$\left(q_{k+1},p_{k+1}\right)$};
  \node[main node] (2) [above right of=4] {$\left(q_{k+1},q_{k+2}\right)$};
  \node[main node] (3) [below left of=1] {$\left(q_k,p_{k}\right)$};
  
  \node[main node] (5) [below right of=2] {$\left(q_{k+2},p_{k+2}\right)$};
  \path[every node/.style={font=\sffamily\small}]
    (1) edge node [left] {$\mathbb{F}^{-}L_d$} (3)
        edge node[right] {$\mathbb{F}^{+}L_d$} (4)
        edge node[above] {${F}_{L_d}$} (2)
    (2) edge node [left] {$\mathbb{F}^{-}L_d$} (4)
        edge node[right] {$\mathbb{F}^{+}L_d$} (5)
    (3) edge node [above] {$\tilde{F}_{L_d}$} (4)
    (4) edge node [above] {$\tilde{F}_{L_d}$} (5);        
    
\end{tikzpicture}
\caption{Flow of the discrete Lagrangian and Hamiltonian map}
\label{fig:comdiag}
\end{figure}
The discrete Hamiltonian map is defined in coordinates as follows:
\begin{align}\label{eq:dis_ham}
\left(q_k,p_k\right) \mapsto \tilde{F}_{L_d}\left(q_k,p_k\right) \defas \left(q_{k+1},p_{k+1} \right) \text{\;\; where \;\;}  \begin{cases} p_k=-D_{1} L_d\left(q_k,q_{k+1}\right),\\p_{k+1}= D_{2} L_d\left(q_k,q_{k+1} \right).\end{cases}
\end{align}

    \subsection{Attitude dynamics in discrete time}    
    We now apply the ideas introduced in \secref{ss:dis_intro} to obtain the discrete equations of the attitude dynamics of a spacecraft. First, the Lagrangian in continuous time is described, and then an approximation of the continuous time Lagrangian is taken to define the discrete Lagrangian \eqref{eq:dis_lag}. Thereafter, the discrete time attitude dynamics is obtained using discrete Hamiltonian formulation \eqref{eq:dis_ham}.

    Consider a rigid body with a point, typically chosen to be center of mass, fixed on it. In order to define the orientation of a rigid body, two coordinate systems are considered with the origin at that fixed point. One frame fixed to the rigid body is known as the body frame, and the other is a frame fixed in space, known as the spatial frame. Let \m{X} be the position of the mass element in the body frame. Then the position of the mass element in reference frame \m{x} is related to the body frame coordinates \m{X} by the rotation matrix \m{R(t) \in \SO{3}} as \m{x(t)=R(t)X}. Let \m{\mathcal{B}} be the region occupied by the body in its reference frame. Let \m{\rho(X)} be the density of the rigid body in the body coordinates at point \m{X}. Then the kinetic energy of the rigid body is ~\cite[p.\ 243]{holm}:
\[ K=\frac{1}{2}\int_{\mathcal{B}} \rho(X) \norm{\dot{x}}^2 d^3 X, \]
which can be rewritten, in view of the left-invariance of the kinetic energy \cite[p.\ 275]{bullo}, as  
\[ K= \frac{1}{2}\int_{\mathcal{B}} \rho(X) \norm{\dot{R}X}^2 d^3 X = \frac{1}{2}\int_{\mathcal{B}} \rho(X) \norm{R^{-1}\dot{R}X}^2 d^3 X.  \]
We know that the spatial angular velocity vector \m{\Omega} can be represented in terms of the body angular velocity vector \m{\omega} as \m{\Omega = R^{-1}\omega.} Then 
\begin{align}\label{eq:rot_mtm}
\Omega \times X= R^{-1}\omega \times R^{-1}x = R^{-1} \left(\omega \times x \right) =  R^{-1}\dot{x}=R^{-1}\dot{R}X.
\end{align}
Let \m{\R^3 \ni \Omega \mapsto \widehat{\Omega} \in \so{3}} be a vector space homeomorphism. Then, from \eqref{eq:rot_mtm} we arrive at the kinematic relation \m{\widehat{\Omega}=R^{-1}\dot{R} \in \so{3},} the Lie algebra of \m{\SO{3}}. So, the kinetic energy can be represented in terms of the spatial frame angular velocity as 
\[ K = \frac{1}{2}\int_{\mathcal{B}} \rho(X)\trace\left((\widehat{\Omega}X)(\widehat{\Omega}X)^\top \right) d^3 X =\frac{1}{2} \trace\left(\widehat{\Omega}\mint \widehat{\Omega}^\top\right), \]
 where \[\mint= \frac{1}{2}\int_{\mathcal{B}} \rho(X) X X^\top d^3 X. \]

    The body moment of inertia matrix \m{J \defas \frac{1}{2}\int_{\mathcal{B}} \rho(X) \widehat{X}^\top \widehat{X} d^3 X } is related to \m{\mint} by the following equation~\cite{taeyoung3d}: 
    \[J=\trace\left(\mint\right)\I - \mint.\] 
If the dissipative and potential forces are absent, then the Lagrangian \m{L:T\SO{3} \rightarrow \R} for the system is given by~\cite[p.\ 245]{holm}
\begin{align*}
L(R,\Omega) \defas K = \frac{1}{2} \trace\left(\widehat{\Omega}\mint \widehat{\Omega}^\top\right). 
\end{align*}
By the kinematic relation \m{\dot{R} = R \widehat{\Omega}} given above, we know that \m{\widehat{\Omega} = R^\top\dot{R}}. So, the Lagrangian can be written as
\begin{align}\label{eq:att_lag}
L(R,\dot{R}) = \frac{1}{2} \trace\left(R^\top\dot{R} \mint \dot{R}^\top R\right).
\end{align}
We now proceed to discretize the Lagrangian \eqref{eq:att_lag}. Considering discrete time instants \m{t_k = kh \text{\;\;for\;\;} k = 0,1,\ldots,} such that \m{R(t_k)=\ornt{k}} and the approximation \m{\dot{R}(t_k) \approx \frac{(\ornt{k+1}-\ornt{k})}{h} \text{\;\;for\;\;} t \in [t_k, t_{k+1}] }~\cite{taeyoung3d}, the discrete Lagrangian \m{L_d : \SO{3}\times \SO{3} \rightarrow \R} is defined as:
\begin{align}
L_d(\ornt{k},\ornt{k+1}) &\approx h L\left(\ornt{k}, \frac{(\ornt{k+1}-\ornt{k})}{h}\right)\nonumber \\
&= \frac{h}{2} \trace \left(\frac{\ornt{k}^\top (\ornt{k+1}-\ornt{k})}{h}\mint \frac{(\ornt{k+1}-\ornt{k})^\top \ornt{k}}{h}\right)\nonumber\\
&= \frac{1}{2h} \trace \left((\ornt{k}^\top \ornt{k+1}-\I)\mint (\ornt{k}^\top \ornt{k+1}-\I)\right)\nonumber\\
&= \frac{1}{h} \trace \left(\left(\I - \rot{k}\right)\mint \right),
\end{align}
where \m{\rot{k} = \ornt{k}^\top \ornt{k+1}}. Note that under the discretization technique employed, the discrete Lagrangian, like its continuous counterpart, is invariant under the action of the SO(3) group. This property will be useful later when momentum equations will be derived and the rotation sequence constructed based on the momentum history.

    Our objective is to come up with first order difference equations describing the attitude dynamics of the spacecraft. By the left trivialization of the cotangent bundle of a Lie group, \m{T^*\SO{3}}  can be represented as \m{\SO{3}\times \so{3}^*}, where \m{\so{3}^*} denotes the dual of the Lie algebra \m{\so{3}}~\cite[p.\ 254]{agrachev}. We now proceed to find the discrete time Hamiltonian map \eqref{eq:dis_ham} \m{\tilde{F}_{L_d}:\SO{3}\times \so{3}^* \rightarrow \SO{3}\times \so{3}^*} such that  
\[\tilde{F}_{L_d}\left(\ornt{k},\widehat{\Pi}_{k}\right) = \left(\ornt{k+1},\widehat{\Pi}_{k+1}\right), \text{\; where\;} \begin{cases} \widehat{\Pi}_{k}=-D_{1} L_d\left(\ornt{k},\ornt{k+1}\right),\\ \widehat{\Pi}_{k+1}= D_{2} L_d\left(\ornt{k},\ornt{k+1} \right), \\ \mtm{k} \in \R^3.\end{cases} \]
In order to find \m{\widehat{\Pi}_{k}} and \m{\widehat{\Pi}_{k+1}},  the variations in \m{\ornt{k}} are defined in terms of \m{\widehat{\eta}_k \in \so{3}}, and the expressions \m{D_{1} L_d\left(\ornt{k},\ornt{k+1}\right)} and \m{D_{2} L_d\left(\ornt{k},\ornt{k+1} \right)} are evaluated using the duality product on \m{\so{3}}~\cite[p.\ 290]{marsden}.\\
\indent For a given \m{\epsilon \in \R} and \m{\eta_k \in \R^{3}}, the variation in \m{\ornt{k}} can be defined as
\begin{align}\label{eq:varR} 
\delta \ornt{k} \defas \left.\frac{d}{d \epsilon}\right|_{\epsilon=0} \ornt{k} \e^{(\epsilon \widehat{\eta}_k)}=\ornt{k} \widehat{\eta}_k, \text{\; where\;} \widehat{\eta}_k \in \so{3};
\end{align} 
then the duality product of \m{\widehat{\Pi}_{k} \in \so{3}^*} and \m{\widehat{\eta}_{k} \in \so{3}} is defined as
\begin{align*}
\frac{1}{2}\trace\left(\widehat{\Pi}_k \widehat{\eta}_k^\top\right)=:\ip{\widehat{\Pi}_k}{\widehat{\eta}_k} &= - \left.\frac{d}{d \epsilon}\right|_{\epsilon=0} L_d(\ornt{k}^\epsilon,\ornt{k+1}) = \frac{1}{h} \trace\left(\left(\delta \ornt{k}^\top \ornt{k+1}\right)\mint \right) \\
&= \frac{1}{h} \trace\left(\widehat{\eta}_k^\top \ornt{k}^\top \ornt{k+1}\mint \right) = \frac{1}{h} \trace\left(\rot{k} \mint \widehat{\eta}_k^\top \right).
\end{align*}
Hence, \[ \trace\left(\underbrace{\left(\frac{1}{2}\widehat{\Pi}_k - \frac{1}{h}\rot{k} \mint\right)}_{C_k}\widehat{\eta}_k^\top\right)=0 \text{\;\; for all\;\;} \widehat{\eta}_k \in \so{3},\]
which means that \m{C_k} is a symmetric matrix. Its skew-symmetric part is, therefore zero, and this leads to
\begin{align}
\widehat{h\Pi}_k = \rot{k} \mint -\mint \rot{k}^\top.
\end{align}
Similarly, the duality product of \m{\widehat{\Pi}_{k+1} \in \so{3}^*} and \m{\widehat{\eta}_{k+1} \in \so{3}} gives,
\begin{align*}
\frac{1}{2}\trace\left(\widehat{\Pi}_{k+1} \widehat{\eta}_{k+1}^\top\right)=:\ip{\widehat{\Pi}_{k+1}}{\widehat{\eta}_{k+1}} &= \left.\frac{d}{d \epsilon}\right|_{\epsilon=0} L_d(\ornt{k},\ornt{k+1}^\epsilon) = \frac{1}{h} \trace\left(-\left(\ornt{k}^\top \delta \ornt{k+1}\right)\mint \right) \\
&= -\frac{1}{h} \trace\left(\ornt{k}^\top \ornt{k+1}\widehat{\eta}_{k+1}\mint \right) = \frac{1}{h} \trace\left(\mint \rot{k} \widehat{\eta}_{k+1}^\top \right),
\end{align*}
and \[ \trace\left(\underbrace{\left(\frac{1}{2}\widehat{\Pi}_{k+1} - \frac{1}{h} \mint \rot{k} \right)}_{D_k}\widehat{\eta}_{k+1}^\top\right)=0 \text{\;\; for all\;\;} \widehat{\eta}_{k+1} \in \so{3}\]
which means that \m{D_k} is a symmetric matrix with its skew-symmetric part equal to zero. Therefore,
\begin{align}
\widehat{\Pi}_{k+1} = \frac{\mint \rot{k} - \rot{k}^\top \mint}{h} = \rot{k}^\top \widehat{\Pi}_k \rot{k} = \widehat{\rot{k}^\top {\Pi}_k},
\end{align}
leading to the following update equation for the momentum:
\begin{align}\label{eq:Mtmwithoutcontol}
\mtm{k+1} = \rot{k}^\top \mtm{k}.
\end{align}
In the presence of control, \eqref{eq:Mtmwithoutcontol} modifies to
\begin{align}\label{eq:MtmUpdt}
\mtm{k+1} = \rot{k}^\top \mtm{k} + h u_{k},
\end{align}
where \m{u_{k}} is the control input at \m{k}th instant of time. The rigid body equations in discrete time are finally obtained as:
\begin{align*}
\begin{matrix}\text{Rigid Body}\\ \text{Dynamics}\end{matrix}&\begin{cases} \ornt{k+1} \nsp&=\ornt{k}\rot{k},\\ \mtm{k+1} \nsp&=\rot{k}^\top \mtm{k} + h \cltr{k},\\ \widehat{h\mtm{k}} \nsp&= \rot{k}\mint - \mint\rot{k}^\top.\end{cases}
\end{align*}

    \section{Optimal control of discrete time attitude dynamics} \label{s:ocdtad}
We state the optimal control problem arising in executing energy optimal attitude manoeuvres of a spacecraft. The spacecraft is assumed to have three actuators, aligned along the three principal moment of inertia axes. Each actuator has its individual saturation limits. The objective is to find the energy optimal control profile for orienting the spacecraft from an initial configuration to a desired configuration in a given duration of time, while obeying pre-specified momentum bounds. First we pose this requirement as an optimal control problem, and then derive the first order necessary optimality conditions using variational analysis. Later \secref{ssec:msmad} the boundary value problem obtained as the first order necessary conditions will be solved using a novel multiple shooting method.
\subsection{Problem description}\label{ssec:PD}
Our objective is to find the energy optimal control law to manoeuvre a spacecraft from the initial configuration \m{(\ornt{i},\mtm{i})} to the final configuration \m{(\ornt{f},\mtm{f})} in \m{N} discrete time steps satisfying the following constraints:
\begin{enumerate}
\item \m{\abs{u_k^i}\leq\cltb^{i} \phantom{space} k=0,1, \ldots,N-1, \quad \text{and} \quad i=1,2,3,}
\item \m{\abs{\mtm{k}^i}\leq\mtmb^{i} \phantom{space} k=1,2, \ldots,N-1, \quad \text{and} \quad i=1,2,3.}
\end{enumerate}
This problem can be posed as an optimal control problem in discrete time as follows:
\begin{align} 
\minimize_{\{\cltr{k}\}}\mathfrak{J}:=\sum_{k=0}^{N-1} \frac{1}{2} \norm{\cltr{k}}^{2}_{2} \label{eq:opt}
\end{align}
subject to
\begin{align}\label{eq:pd_mdl}
&\text{system of equations}\begin{cases} \ornt{k+1}\nsp&=\ornt{k}\rot{k}\\ \mtm{k+1}\nsp&=\rot{k}^\top\mtm{k} + h \cltr{k}\\ \widehat{h\mtm{k}}\nsp&= \rot{k}\mint - \mint\rot{k}^\top\end{cases} \phantom{sp}\text{with}\\
&\text{boundary conditions} \quad (\ornt{0},\mtm{0})=(\ornt{i},\mtm{i}), (\ornt{N},\mtm{N})= (\ornt{f},\mtm{f}), \text{\; and\;} \label{eq:pd_bnd}\\
& \text{constraints\;\;}\begin{cases} \left(\cltr{k}^i\right)^2\leq\left(\cltb^{i}\right)^2 \quad \text{for all} \quad k=0,1,\ldots,N-1, \quad \text{and} \quad i=1,2,3,\\ \left(\mtm{k}^i\right)^2\leq \left(\mtmb^{i}\right)^2 \quad \text{for all}\quad k=1,2, \ldots,N-1, \quad \text{and}\quad i=1,2,3.\end{cases} \label{eq:pd_inq}
\end{align}
Note that the optimal control problem \eqref{eq:opt} has both control and state inequality constraints. While the individual control inputs are constrained in magnitude, the performance measure reflects a 2-norm on the control action at each stage.   
    \subsection{Necessary optimality conditions}
We represent the variations of \m{\rot{k}} in terms of variations in \m{\mtm{k}} and then the first order necessary conditions are derived. 
\begin{itemize}
\item \textbf{Representation of the variations:}\\
Using \eqref{eq:varR} the variations for the matrix \m{\ornt{k}^\top \ornt{k+1}} is defined as
\begin{align}\label{eq:rrk}
\delta\left(\ornt{k}^\top \ornt{k+1}\right)=\delta \ornt{k}^\top \ornt{k+1} + \ornt{k}^\top \delta \ornt{k+1}=-\widehat{\eta_k} \ornt{k} + \ornt{k} \widehat{\eta}_{k+1}.
\end{align}
Using the property \m{\widehat{F^\top x} = F^\top \widehat{x} F}, \eqref{eq:rrk} simplifies to 
\begin{align}
\delta\left(\ornt{k}^\top \ornt{k+1}\right)=\rot{k} \left(-\rot{k}^\top \eta_k + \eta_{k+1}\right)^{\wedge},\label{eq:ornt_var}
\end{align}
where \m{(\cdot)^{\wedge}:\R^3 \rightarrow \so{3}.}
Similarly, for a given \m{\xi_k \in \R^3,} we define the variation in \m{\rot{k} \in \SO{3}} as 
\begin{align}\label{eq:rot_var}
\delta \rot{k}=\rot{k}\widehat{\xi}_k.
\end{align}
The implicit equation \m{\widehat{h\mtm{k}} = \rot{k}\mint - \mint\rot{k}^\top} in \eqref{eq:pd_mdl} gives the relation between momentum and change in orientation at \m{k}th time instant. So, the relation between the variations in momentum \m{\delta \mtm{k}} and \m{\delta \rot{k}}  can be obtained from the implicit equation as~\cite{leok3d}
\begin{align*}
\widehat{h\delta\mtm{k}} = \delta\rot{k}\mint - \mint\delta\rot{k}^\top =\widehat{\rot{k}\xi_k} \rot{k} \mint + \mint \rot{k}^\top \widehat{\rot{k} \xi_k},
\end{align*}
which can be further simplified using \m{\widehat{x}A+A^\top\widehat{x}=\left(\left\{\trace[A]\I-A\right\}x\right)^\wedge} to
\begin{align*}
\widehat{h\delta\mtm{k}}=\left(\mtrt{k} \rot{k} \xi_k \right)^\wedge.
\end{align*}
\begin{restatable}{lemma}{invFkJd}
\label{lemma:invFkJd}
The matrix \m{\mtrt{k}} is invertible if \[\cos\left(\frac{\norm{\xi_k}}{2}\right) < \sqrt{\frac{2 d_3 + d_2 -d_1}{2(d_3 + d_2)}},\]
where \m{\rot{k} = \e^{\widehat{\xi}_k},\quad \xi_k \in \R^3}, and \m{\mint = diag(d_1,d_2,d_3) \text{\;such that\;} 0 < d_1 \leq d_2 \leq d_3.}
\end{restatable}
We present a proof of Lemma \ref{lemma:invFkJd} in Appendix \ref{app:invFkJd}. Armed with Lemma \ref{lemma:invFkJd}, we represent the vector \m{\xi_k} in terms of the variations in momentum i.e. \m{\delta \mtm{k}} ~\cite{leok3d} as
\begin{align}\label{eq:implicit_var}
\xi_k = \mathcal{B}_k \delta \mtm{k},
\end{align}
where \m{\mathcal{B}_k = h \rot{k}^\top \mtrt{k}^{-1}.} 
\item \textbf{Necessary optimality conditions:}\\
Let \m{\cor{k}\in \R^3} and \m{\com{k}\in \R^3} be the Lagrange multipliers corresponding to the equality constraints \m{\ornt{k+1}-\ornt{k}\rot{k}=0} and \m{\mtm{k+1}-\rot{k}^\top \mtm{k} - h \cltr{k}=0} respectively. Similarly let \m{0 \leq \clts{k}^i \in \R} and \m{0 \leq \mtms{k}^i \in \R} be the Lagrange multipliers corresponding to the inequality constraints \m{\left(\cltr{k}^i\right)^2\leq\left(\cltb^{i}\right)^2 \text{\; and \;} \left(\mtm{k}^i\right)^2\leq \left(\mtmb^{i}\right)^2 }. Let us justify why the Lagrange multiplier \m{\cor{k}\in \R^3} is chosen corresponding  to the rotational kinematics \m{\ornt{k+1}-\ornt{k}\rot{k}=0}. Rotational kinematics can be rewritten as \m{\ornt{k}^\top\ornt{k+1}-\rot{k}=0}, where \m{\rot{k}} can be identified by its skew symmetric part, which in turn can identified by a vector in \m{\R^3}~\cite{leokattitude}.
\begin{restatable}{claim}{fskew}
\label{claim:fskew}
Consider the equality \m{\ornt{k}^\top \ornt{k+1} = \rot{k}}. If we assume that the step length \m{h} is small enough such that the relative orientation \m{\ornt{k}^\top\ornt{k+1}} between two adjacent time instances \m{t_k} and \m{t_{k+1}} is less than \m{\frac{\pi}{2},} i.e.,
 \[\norm{\xi_k} ,\norm{\zeta_k} < \frac{\pi}{2} \text{\;where\;} \ornt{k}^\top\ornt{k+1} = \e^{\widehat{\zeta}_k},\rot{k} = \e^{\widehat{\xi}_k}\; and \; \zeta_k, \xi_k \in \R^3,\] then the aforementioned equality is satisfied if and only if the skew symmetric parts of both sides are identical.
\end{restatable}
 The augmented performance index can be defined as
\begin{align*}
\minimize_{\abs{\cltr{k}^i}\leq \cltb^{i}}\mathfrak{J}_a:=&\sum_{k=0}^{N-1}\frac{1}{2}\ip{\cltr{k}}{\cltr{k}} + \ip{\com{k}}{-\mtm{k+1}+\rot{k}^\top \mtm{k}+ h \cltr{k}}\\
& + \ip{\cor{k}}{\frac{1}{2}\left(\rot{k}-\rot{k}^\top\right)^\vee-\frac{1}{2}\left(\ornt{k}^\top \ornt{k+1}-\ornt{k+1}^\top \ornt{k}\right)^\vee}\\
& + \frac{1}{2} \ip{\clts{k}}{\cltr{k} \odot \cltr{k} - \cltb \odot \cltb} +\sum_{l=1}^{N-1}\frac{1}{2} \ip{\mtms{l}}{\mtm{l} \odot \mtm{l} - \mtmb \odot \mtmb}, 
\end{align*}
where
\[\cltb= \rvec{\cltb}^\top, \quad \mtmb= \rvec{\mtmb}^\top, \quad x\odot y = \begin{pmatrix}x^1y^1 & x^2y^2 & x^3y^3\end{pmatrix}^\top. \]

Using \eqref{eq:ornt_var}, \eqref{eq:rot_var} and \eqref{eq:implicit_var}, the infinitesimal variation of the augmented performance index is defined as
\begin{align*}
\delta\mathfrak{J}_a:=&\sum_{k=0}^{N-1}\{\ip{\delta\cltr{k}}{\cltr{k}} + \ip{\frac{\cor{k}}{2}}{\delta\left(\rot{k}-\rot{k}^\top\right)^\vee-\delta\left(\ornt{k}^\top \ornt{k+1}-\ornt{k+1}^\top \ornt{k}\right)^\vee}\\
& + \ip{\com{k}}{-\delta\mtm{k+1}+\delta\rot{k}^\top \mtm{k} + \rot{k}^\top\delta\mtm{k}+h \delta \cltr{k}} + \ip{\delta \cltr{k}}{\clts{k} \odot \cltr{k}}\}\\
& +\sum_{l=1}^{N-1} \ip{\delta\mtm{l}}{\mtms{l} \odot \mtm{l}},
\end{align*}
where 
\begin{align}\label{eq:nneg} 
\clts{k}^i,\mtms{l}^i  \geq 0 \quad \text{for all \;} i=1,2,3,
\end{align}
together with the complementary slackness conditions
\begin{align}\label{eq:comslack}
\mtms{l}^i\left((\mtm{l}^i)^2 - (\mtmb^{i})^2\right)=0 \quad \text{and} \quad \clts{k}^i\left((u_{k}^i)^2 - (\cltb^{i})^2\right)=0.
\end{align}
Employing the property that \m{\widehat{x}A + A^\top \widehat{x} = \left\{\left(\trace[A]\I-A\right)x\right\}^\wedge} and that the variation on the boundary is zero, i.e., \m{\eta_0 = \eta_N =0 , \delta\mtm{0}=\delta\mtm{N} =0}, we rearrange the terms of the expression \m{\delta\mathfrak{J}_a\;}and get the following expression after standard algebraic manipulations:
\begin{align}
\delta\mathfrak{J}_a & = \sum_{k=0}^{N-1}\ip{\delta\cltr{k}}{\cltr{k}+h\com{k}+\clts{k} \odot \cltr{k}} \nonumber \\
&+ \sum_{k=1}^{N-1} \ip{\frac{\eta_k}{2}}{\rot{k}\tr{k}\cor{k}-\tr{k-1} \cor{k-1}} \nonumber\\
&+ \sum_{k=1}^{N-1} \ip{\delta\mtm{k}}{-\com{k-1}+\left(\rot{k}-\mathcal{B}_k^\top \widehat{\rot{k}^\top\mtm{k}}\right)\com{k} }\nonumber\\
&+ \sum_{k=1}^{N-1} \ip{\delta\mtm{k}}{\mtms{k} \odot \mtm{k}+\frac{1}{2}\mathcal{B}_k^\top \tr{k}\cor{k}}.\label{eq:acv}
\end{align}
By first order necessary condition of optimality, \m{\delta\mathfrak{J}_a = 0} along all possible variations \m{\delta \cltr{k}, \eta_k, \delta\mtm{k}}, and we obtain the co-state equations from \eqref{eq:acv} as:
\begin{align}\label{eq:costate}
\begin{matrix}\text{Co-state}\\ \text{Dynamics}\end{matrix} 
\begin{cases}
\tr{k-1} \cor{k-1}-\rot{k}\tr{k}\cor{k}=0,\\
\mtms{k} \odot \mtm{k} - \com{k-1}+\left(\rot{k}-\mathcal{B}_k^\top \widehat{\rot{k}^\top\mtm{k}}\right)\com{k} +\frac{1}{2}\mathcal{B}_k^\top \tr{k}\cor{k} =0,
\end{cases}
\end{align}
and the optimality condition for the control as
\begin{align}\label{eq:cltr}
\ocltr{k}+h\com{k}+\clts{k} \odot \ocltr{k} = 0. 
\end{align}
From \eqref{eq:comslack} we know that if \m{\ocltr{k}^i } lies in the interior of the feasible region \m{\mathcal{C}^i \defas \left\{u^i \in \R \left| (u^i)^2 < (c^i)^2 \right. \right\}}, then  \m{\clts{k}^i =0}. So, by \eqref{eq:cltr} we have \m{h\com{k}^i=-\ocltr{k}}. On the other hand if \m{\ocltr{k}^i = \cltb^{i}} then by \eqref{eq:cltr} and \eqref{eq:nneg} we have \m{h\com{k}^i= -\cltb^{i} (1+\clts{k}^i)\leq- \cltb^{i}}. Similarly for \m{\ocltr{k}^i = -\cltb^{i}}, we have \m{h\com{k}^i =\cltb^{i} (1+\clts{k}^i)\geq \cltb^{i}}. Hence, the optimal control can be written in a compact form as 
\begin{align}\label{eq:ocltr}
\ocltr{k}^i = - \min \left\{c^i, \abs{h\com{k}^i}\right\} \sgn (\com{k}^i).
\end{align}
\end{itemize}

	We represent the state and co-state dynamics in terms on momentum and co-state corresponding to momentum variables. This technique will reduce the model of the system and hence the algorithm will perform better in terms of memory and  time requirement.
 
\section{Scaling and model reduction}\label{s:smr}
The orientation boundary constraints \m{\left(\ornt{0},\ornt{N}\right) = \left(\ornt{i},\ornt{f}\right)\;} can be represented in terms of \m{\rot{k},} which in turn can be computed for a given \m{\mtm{k}} using the implicit form  
\begin{align}
 \widehat{h\mtm{k}}= \rot{k}\mint - \mint\rot{k}^\top. \label{eq:imp}
\end{align}
We represent orientation constraints first in terms of \m{\rot{k}} and then we discuss a technique for computing \m{\rot{k}} for a given value of \m{\mtm{k}}.
\subsection{Representing orientation constraints in terms of momentum} 
We represent the boundary constraints on orientation \m{\left(\ornt{0},\ornt{N}\right) = \left(\ornt{i},\ornt{f}\right)\;} \eqref{eq:pd_bnd} in terms of \m{\rot{k}}. Then using \eqref{eq:pd_mdl} we reconstruct the orientation \m{\ornt{k}} variables, we see that \m{\ornt{f}} can be represented in terms of \m{\rot{k}} for \m{k=1,2, \ldots, N-1,} as
%\begin{align}
\[\ornt{f}=\ornt{N} = \ornt{i}\rot{0}\rot{1}\rot{2} \ldots \rot{N-1}.\]
%\end{align}
The boundary constraints of orientation can be rewritten as:
\begin{align}\label{eq:ornt_bnd}
\ornt{i}^\top \ornt{f} = \rot{0}\rot{1}\rot{2} \ldots \rot{N-1}.
\end{align}
To represent the constraints \eqref{eq:ornt_bnd} in vector form, we have nine equations. From \eqref{eq:costate} we know that the number of free variables corresponding to orientation kinematics are actually three: \m{\cor{0} \in \R^3}. So, the boundary conditions \eqref{eq:ornt_bnd} have to be represented by three independent constraints. 
Let us define the maps \m{\SO{3} \ni M \mapsto \logm(M) \in \so{3}} and \m{\so{3} \ni x \mapsto (x)^{\vee} \in \R^3}. Then the boundary condition \eqref{eq:ornt_bnd} is satisfied if 
\begin{align}\label{eq:ocst}
\text{Orientation constraints:\;}\quad \Orntcnst  \defas \left(\logm\left(\ornt{f}^\top \ornt{i} \rot{0}\rot{1}\rot{2} \ldots \rot{N-1} \right) \right)^{\vee} = 0.
\end{align}
We need to compute the gradient of the orientation constraints \eqref{eq:ocst} w.r.t. momentum \m{\mtm{k}}. This can be done as follows:
From \eqref{eq:ocst} we know that
\[\e^{\Orntcnst^\wedge} = \ornt{f}^\top \ornt{i} \rot{0}\rot{1}\ldots\rot{k} \ldots \rot{N-1}, \]
leading to
\[\e^{\Orntcnst^\wedge}\der{\mtm{k}^i}\Orntcnst^\wedge = \ornt{f}^\top \ornt{i} \rot{0}\rot{1}\ldots \der{\mtm{k}^i}(\rot{k}) \ldots \rot{N-1}. \]
 After algebraic manipulations we get
\[\der{\mtm{k}^i}\Orntcnst^\wedge =  \rot{N-1}^\top \rot{N-2}^\top \ldots \rot{k}^\top \der{\mtm{k}^i}(\rot{k}) \rot{k+1}\ldots  \rot{N-1},\]
and using the property \m{\widehat{A^\top x} = A^\top \hat{x} A } we conclude that
\[ \der{\mtm{k}^i}\Orntcnst = \rot{N-1}^\top \rot{N-2}^\top \ldots \left(\rot{k}^\top \der{\mtm{k}^i}\rot{k}\right)^{\vee}.\]
Therefore, the gradient of the orientation constraints \eqref{eq:ocst} w.r.t. momentum \m{\mtm{k}} is
\begin{align}\label{eq:docst}
\der{\mtm{k}}\Orntcnst=\begin{pmatrix} \der{\mtm{k}^1}\Orntcnst & \der{\mtm{k}^2}\Orntcnst & \der{\mtm{k}^3}\Orntcnst. \end{pmatrix}
\end{align}

We notice that the orientation constraints \eqref{eq:ocst} can be computed only when we construct the matrix \m{\rot{k}} for a given value of the momentum vector \m{\mtm{k}}. Similarly, the gradient of the orientation constraints \eqref{eq:docst} can be computed only when we find the derivative of \m{\rot{k}} w.r.t. \m{\mtm{k}^i} for \m{k=1,2,\ldots, N-1,} and \m{i=1,2,3.}

\subsection{Determining \m{\rot{k}} and \m{\drot{k}} in terms of momentum \m{\mtm{k}}}\label{ssec:rotmtm}
In order to calculate the orientation constraints \eqref{eq:ocst}, \m{\rot{k}} is obtained in terms of \m{\mtm{k}}. We know that \m{\rot{k}} can be obtained from \m{\mtm{k}} by solving the implicit equation \eqref{eq:imp}. To solve this implicit form, we choose quaternions to parameterize the matrix \m{\rot{k}}. 
Let 
\[\rot{k} \defas \begin{pmatrix} q_0^2+q_1^2-q_2^2-q_3^2 & 2 q_1 q_2-2 q_0 q_3 & 2 q_1 q_3 + 2 q_0 q_2\\ 2 q_1 q_2+2 q_0 q_3 &  q_0^2-q_1^2+q_2^2-q_3^2 &  2 q_2 q_3 - 2 q_0 q_1\\ 2 q_1 q_3 - 2 q_0 q_2 & 2 q_2 q_3 + 2 q_0 q_1 & q_0^2-q_1^2-q_2^2+q_3^2\end{pmatrix}, \]
 and the inertia matrix \m{\inertia} and \m{\mint} are defined as:
\[ \inertia \defas \begin{pmatrix} \ii{x} &0 &0\\0 & \ii{y} & 0\\ 0 & 0 & \ii{z} \end{pmatrix}, \phantom{spa} \mint \defas \frac{1}{2}\begin{pmatrix} -\ii{x}+\ii{y}+\ii{z} &0 &0\\0 & \ii{x}-\ii{y}+\ii{z} & 0\\ 0 & 0 & \ii{x}+\ii{y}-\ii{z} \end{pmatrix}.\]
Let \m{\mtm{k} \defas\begin{pmatrix} \mtm{k}^1 & \mtm{k}^2 & \mtm{k}^3 \end{pmatrix}^\top} be the momentum of the body at the \m{k}th instant, and \m{q=\left(q_0,q_1,q_2,q_3\right)^\top}; then \eqref{eq:imp} can be represented in the form of nonlinear algebraic equations as follows:
\begin{align}\label{eq:algqm}
g\left(q(\mtm{k}),\mtm{k}\right) \defas
\begin{pmatrix} 2 q_2 q_3 \left(\ii{z}-\ii{y}\right) + 2 q_0 q_1 \ii{x} -h \mtm{k}^1 \\ 
          2 q_1 q_3 \left(\ii{x}-\ii{z}\right) + 2 q_0 q_2 \ii{y} -h \mtm{k}^2 \\ 
              2 q_1 q_2 \left(\ii{y}-\ii{x}\right) + 2 q_0 q_3 \ii{z} - h \mtm{k}^3 \\ 
          q_0^2+q_1^2+q_2^2+q_3^2 -1
\end{pmatrix} = 0.	
\end{align}
For a fixed value of momentum \m{\mtm{k}}, the system of equation \eqref{eq:algqm} has quaternions as unknown parameters which can be numerically found using Newton's method at each instant of time. 
Since \m{\rot{k}} is represented in terms of quaternions, we first need to find the variations of the quaternions in terms of the momentum vector \m{\mtm{k}}; these can be obtained by taking derivative of \eqref{eq:algqm} w.r.t. \m{\mtm{k}} as 
\[\der{q} g\left(q(\mtm{k}),\mtm{k}\right) \der{\mtm{k}} q(\mtm{k})+ \der{\mtm{k}} g\left(q(\mtm{k}),\mtm{k}\right) =0, \]
resulting in
\begin{align}\label{eq:dqtnmtm}
\der{q} g\left(q(\mtm{k}),\mtm{k}\right)
&\underbrace{\begin{pmatrix} \frac{\partial q_0}{\partial \mtm{k}^1} & \frac{\partial q_0}{\partial \mtm{k}^2} & \frac{\partial q_0}{\partial \mtm{k}^3}\\ \frac{\partial q_1}{\partial \mtm{k}^1} & \frac{\partial q_1}{\partial \mtm{k}^2} & \frac{\partial q_1}{\partial \mtm{k}^3}\\ 
\frac{\partial q_2}{\partial \mtm{k}^1} & \frac{\partial q_2}{\partial \mtm{k}^2} & \frac{\partial q_2}{\partial \mtm{k}^3}\\ 
\frac{\partial q_3}{\partial \mtm{k}^1} & \frac{\partial q_3}{\partial \mtm{k}^2} & \frac{\partial q_3}{\partial \mtm{k}^3}
\end{pmatrix}}_{\der{\mtm{k}}q(\mtm{k})}
= \begin{pmatrix} h& 0& 0\\0& h& 0\\ 0& 0& h\\ 0& 0& 0 \end{pmatrix},
\end{align}
where
\begin{align*}
\der{q} g\left(q(\mtm{k}),\mtm{k}\right) =
\begin{pmatrix}2 q_1 \ii{x} & 2 q_0 \ii{x} & 2 q_3 \left(\ii{z}-\ii{y}\right) & 2 q_2 \left(\ii{z}-\ii{y}\right) \\
           2 q_2 \ii{y} & 2 q_3 \left(\ii{x}-\ii{z}\right) & 2 q_0 \ii{y} & 2 q_1 \left(\ii{x}-\ii{z}\right)\\
           2 q_3 \ii{z} & 2 q_2 \left(\ii{y}-\ii{x}\right) & 2 q_1 \left(\ii{y}-\ii{x}\right) & 2 q_0 \ii{z}\\
           2 q_0 & 2 q_1 & 2 q_2 & 2 q_3 
\end{pmatrix}. 
\end{align*}
The matrix \m{\der{\mtm{k}}q} can be obtained by solving the linear system \eqref{eq:dqtnmtm}. The derivative of the matrix \m{\rot{k}} w.r.t. \m{\mtm{k}^i} is obtained using the chain rule as
\begin{align}
\der{\mtm{k}^i} \rot{k} = \sum_{n=0}^{3} \der{q_n}\rot{k} \frac{\partial q_n}{\partial \mtm{k}^i}.
\end{align}
We now discuss the reduction of the difference equation model \eqref{eq:pd_mdl}, \eqref{eq:costate} to the momentum and comomentum dynamics respectively and later the reduced model is scaled by change of variables so as to make the model invariant under change in the step length \m{h}. Invariance of the difference equation model means that for a particular manoeuvre the optimal trajectory and corresponding Lagrange multipliers remain identical for different step lengths \m{h}. This matter is quite essential because it largely affects the region of convergence and order of convergence of the algorithm~\cite{gill}. 

\subsection{Scaling and model reduction}
First we discuss about model reduction and later scaling of the reduced model by appropriate change of variable. Let us define a new variable \m{\mcor{k} \defas 2h^{-2} \tr{k} \cor{k}} then \eqref{eq:costate} can be written as 
\begin{align} \label{eq:mcsysm}
\mtms{k} \odot \mtm{k} -\com{k-1} + \left(\rot{k}-h\mathcal{N}_k \widehat{\rot{k}^\top\mtm{k}}\right)\com{k} + \frac{\mathcal{N}_k}{h} \mcor{k} = 0\\
\mcor{k-1} -\rot{k} \mcor{k} =0 \label{eq:mcsysr}
\end{align}
where \m{\mathcal{N}_k = \frac{\mathcal{B}_k^\top}{h}}. From \eqref{eq:mcsysr} we conclude \m{\mcor{k} = Q_k^\top \mcor{0}} such that \m{Q_k = \rot{1}\rot{2}\ldots\rot{k}}. So \eqref{eq:mcsysm} and \eqref{eq:pd_mdl} can be further reduced to the system of difference equations 
\begin{align} \label{eq:sde}
&\begin{pmatrix}
\mtm{k+1}-\rot{k}^\top\mtm{k} - h \ocltr{k}\\
\mtms{k+1} \odot \mtm{k+1}-\com{k} +\left(\rot{k+1}-h\mathcal{N}_{k+1} \widehat{\rot{k+1}^\top\mtm{k+1}}\right)\com{k+1} +  \mathcal{N}_{k+1} \frac{Q_{k+1}^\top}{h}\mcor{0}\end{pmatrix} =0.
\end{align}
Assume for the moment that for a manoeuvre the \m{i}th element of the control vector \m{\ocltr{k}} saturates (i.e., \m{\abs{\ocltr{k}^i} = \cltb^{i}}) at the \m{k}th instant of time. Clearly, from \eqref{eq:ocltr} \m{\com{k}^i \geq \frac{\cltb^i}{h}.} If the step length \m{h} is small, \m{\com{k}^i} will take very large values, which makes the difference equations \eqref{eq:sde} stiff. To avoid this situation, we define a variable \m{\mcom{k} \defas h\com{k}}, and modify the difference equation \eqref{eq:sde} to 
\begin{align} \label{eq:rsde}
&\begin{pmatrix}
\mtm{k+1}-\rot{k}^\top\mtm{k} - h \ocltr{k}\\
h\mtms{k+1} \odot \mtm{k+1}-\mcom{k} +\left(\rot{k+1}-h\mathcal{N}_{k+1} \widehat{\rot{k+1}^\top\mtm{k+1}}\right)\mcom{k+1} +  \mathcal{N}_{k+1} Q_{k+1}^\top\mcor{0}\end{pmatrix} =0,
\end{align}
where 
\begin{align}\label{eq:mocltr}
\ocltr{k}^i &= -\min \left\{c^i, \abs{\mcom{k}^i}\right\} \sgn (\mcom{k}^i),\quad \mathcal{N}_k &= \mtr{k}^{-1} \rot{k}. 
%\Xi_{k+1} &= \mcom{k+1} - \left(\rot{k+1}-h\mathcal{N}_{k+1} \widehat{\rot{k+1}^\top\mtm{k+1}}\right)^{-1}\left(\mcom{k} - \mathcal{N}_{k+1} Q_{k+1}^\top \mcor{0}\right). \nonumber
\end{align}
Note that \m{\rot{k}} and its gradient with respect to momentum can be obtained as discussed in \secref{ssec:rotmtm}. In the following section we employ multiple shooting method to solve the system of difference equations \eqref{eq:rsde} with boundary conditions \eqref{eq:pd_bnd} and constraints \eqref{eq:pd_inq}.    

    \section{Multiple Shooting Methods} \label{s:msm}
Shooting methods were mainly developed for solving ordinary differential equations or difference equations with given boundary conditions. An initial guess is taken for the unknown initial values of the differential or difference equation variables. Then the variables are computed at the terminal time and compared with a known value of the variables at the boundaries. Then the initial guess is improved at each iteration to match with the known boundary values. The multiple shooting methods, the time domain is divided into sub-intervals (time domain decomposition), and the boundary value problems are solved for each sub-interval with the condition that the boundary values at the common points of the adjacent intervals are the same. Multiple shooting methods are generalizations of the single shooting method in the sense that  multiple shooting with a single time interval is equivalent to the shooting method~\cite{helkePhd08}. \\
\phantom{space} Multiple shooting methods have many advantages over single shooting methods. The former are more stable, and hence can be applied to stiff problems, and the time domain decomposition allows one to introduce the initial guess to the problem with prior knowledge. Multiple shooting methods allow one to compute the solution of the differential equation at individual intervals, which can be very efficient in computation using parallel architecture~\cite{helkePhd08}.

\subsection{A quick introduction to multiple shooting methods}
 Multiple shooting methods constitute generalizations of the single shooting method, in which two point boundary value problems are solved at each iteration for the subintervals of time domain simultaneously. Let 
\begin{align}\label{eq:bvp}
&\dot{x}=f(x,t) \phantom{space} x \in \R^{n} \quad \text{with boundary conditions}\\
&B:\R^n \times \R^n \rightarrow \R^n \phantom{sp} \text{such that\;} B(x(a),x(b)) = 0.
\end{align}
Let the time domain be decomposed into \m{N} sub-interval as
\[ t_0=\tau_{0}<\tau_{1} \ldots \tau_{N-1}<\tau_{N}=t_f,\]
and let us consider \m{(N+1)} variables \m{s^0,s^1,\ldots,s^n,} known as the multiple shooting variables. These multiple shooting variables are the guessed initial values of the dependent variable \m{x} defined in \eqref{eq:bvp} at the specified time instants. We now define initial value problems for each sub-interval as
\begin{align}\label{eq:bvpi}
&\dot{x}^k=f(x^k,t) \quad \text{such that\;}x^k:[\tau_{k}, \tau_{k+1}] \rightarrow \R^{n} \text{\;with}\\
& x^{k}(\tau_{k})=s^{k} \text{\; for\;} k=0,1,\ldots,(N-1).
\end{align}
Note that the solution \m{\left\{x^k(\cdot,s^{k})| k=0,1, \ldots,(N-1)\right\}} to the initial value problems \eqref{eq:bvpi} can be a solution to the boundary value problem \eqref{eq:bvp} only if the solution \m{x^k} of the interval \m{[\tau_{k}, \tau_{k+1}]} matches with the initial condition for the next interval i.e. \m{x^k(\tau_{k+1},s^k)=s^{k+1}}. This condition is known as the matching condition. 
These matching conditions for each interval can be combined together and can be represented in stacked form as
\begin{align}\label{eq:mcds}
F(s) \defas \begin{pmatrix} s^1-x^0(\tau_{1},s^0)\\s^2-x^1(\tau_{2},s^1)\\ \vdots \\s^N-x^{N-1}(\tau_{N},s^{N-1})\\ B(x(a),x(b)) \end{pmatrix} =0, 
\end{align}
with \m{x^k(\cdot,s^k)} the solution of initial value problem \eqref{eq:bvpi}, and \m{s \defas \left(s^0,s^1,\ldots,s^N\right)}. The algebraic equations \eqref{eq:mcds} can be solved using Newton type algorithms for multi-variable functions is described in Algorithm \ref{alg:msm}. 
\begin{algorithm}[H]
\caption{Multiple shooting method}\label{alg:msm}
\begin{algorithmic}[1]
\Procedure{Root of the function \m{F}}{}
\State Choose appropriate initial guess \m{s} and the error tolerance bound \m{\delta>0}.
\State \textbf{Stopping Criterion:}
\If {\m{\norm{F(s)}_{\infty} < \delta}} Stop.  
\EndIf
\State Calculate \m{\Delta s} by solving the linear system \m{\der{s}F(s)\Delta s = - F(s)}. 
\State Update the value of \m{s=s+\Delta s.} 
\State \textbf{goto} \textbf{Stopping Criterion}.
\EndProcedure
\end{algorithmic}
\end{algorithm}

\subsection{Multiple shooting method for attitude dynamics} \label{ssec:msmad}
The system dynamics \eqref{eq:rsde} with the boundary conditions \eqref{eq:ocst} and \m{\left(\mtm{0},\mtm{N} \right) = \left(\mtm{i},\mtm{f} \right)} have to be solved with additional inequality constraints \m{(\mtm{k}^i)^2 \leq (\mtmb^i)^2} along with the complimentary slackness conditions \eqref{eq:comslack} defined by
\begin{align} \label{eq:sys}
&\begin{cases}
\mtm{k+1}-\rot{k}^\top\mtm{k} - h \ocltr{k}=0,\\
h\mtms{k+1} \odot \mtm{k+1}-\mcom{k} +\left(\rot{k+1}-h\mathcal{N}_{k+1} \widehat{\rot{k+1}^\top\mtm{k+1}}\right)\mcom{k+1} +  \mathcal{N}_{k+1} Q_{k+1}^\top\mcor{0}=0, \end{cases}\\
&\text{boundary conditions} \begin{cases} \mtm{N} - \mtm{f} = 0, \quad \mtm{0} - \mtm{i} = 0, \quad \Orntcnst = 0, \end{cases} \label{eq:bnd} \\
&\text{slackness conditions} \begin{cases} \mtms{k}^i\left((\mtm{k}^i)^2 - (\mtmb^{i})^2\right)=0, \quad \mtms{k}^i \geq 0, \quad (\mtm{k}^i)^2 \leq (\mtmb^i)^2,\end{cases} \label{eq:inq}
\end{align}
where 
\begin{align*}
&\ocltr{k}^i = - \min \left\{c^i, \abs{\mcom{k}^i}\right\} \sgn (\mcom{k}^i), \quad \mathcal{N}_k = \mtr{k}^{-1} \rot{k},\\
&\Orntcnst \defas \left(\logm\left(\ornt{f}^\top \ornt{i} \rot{0}\rot{1}\rot{2} \ldots \rot{N-1} \right) \right)^{\vee}.
\end{align*}
 First we consider the case in which state constraints are not active (i.e., \m{\mtms{k}=0 \;} for all \m{k}). Then \eqref{eq:inq} is trivially satisfied. If state constraints are not active, then the necessary conditions are defined by \eqref{eq:sys} with \m{\mtms{k}=0 \;} for all \m{k}, and \eqref{eq:bnd}. We now represent the matching and boundary condition for the system of difference equations \eqref{eq:sys} and \eqref{eq:bnd} in terms of \m{\mtm{k}, \mcom{k}} and \m{\mcor{0}} and solve the system of nonlinear algebraic equations comprises of matching conditions.

\subsubsection{Matching conditions for multiple shooting methods}
The system of difference equations \eqref{eq:sys} along with the boundary conditions \eqref{eq:bnd} will result in the following set of matching conditions, which can be solved using Newton's root finding algorithm with quadratic convergence rate~\cite{polyaknewton}.\\
Solve:
\begin{align}\label{eq:matching}
\mathcal{M}(X)\defas 
\begin{pmatrix}
\Sigma_1 & \Xi_1 & \ldots& \Sigma_k& \Xi_k&  \ldots& \Sigma_{N}& \Xi_{N}&
\Mtmcnst^\top & \Orntcnst^\top \end{pmatrix}^\top = 0,
\end{align}
where 
\begin{align*}
&X \defas \begin{pmatrix}\mtm{0}^\top & \mcom{0}^\top &\ldots &\mtm{k}^\top &\mcom{k}^\top &\ldots &\mtm{N}^\top &\mcom{N}^\top &\mcor{0}^\top \end{pmatrix}^\top,\\
&\Sigma_{k+1} \defas \left(\mtm{k+1}-\rot{k}^\top\mtm{k} - h \ocltr{k}\right)^\top, \\
& \Xi_{k+1} \defas \left(\mathcal{N}_{k+1} Q_{k+1}^\top \mcor{0} + \left(\rot{k+1}-h\mathcal{N}_{k+1} \widehat{\rot{k+1}^\top \mtm{k+1}}\right) \mcom{k+1} -\mcom{k}\right)^\top , \\
& \Mtmcnst \defas \begin{pmatrix} \mtm{N} - \mtm{f} \\ \mtm{0} - \mtm{i} \end{pmatrix}, \quad \Orntcnst \defas \left( \logm \left( \ornt{f}^\top \ornt{i} \rot{0}\rot{1}\rot{2} \ldots \rot{N-1} \right) \right)^{\vee}. 
\end{align*}
Let \m{\der{X}\mathcal{M}} be the Jacobian matrix of the matching conditions \m{\mathcal{M}} \eqref{eq:matching}. Assuming \m{X_n} to be the solution of the system \m{\mathcal{M}} at the \m{n}th iteration of the Newton's root finding algorithm, the \m{(n+1)}th iteration is given by
\[X_{n+1} = X_n + \Delta X_n, \] 
where \m{\Delta X_n } is a solution of the linear system \m{\der{X}\mathcal{M}(X_n)\Delta X_n = -\mathcal{M}(X_n).}

	It is important to note that Newton's update \m{\Delta X_n} is not necessarily unique; \m{\Delta X_n} is unique if and only  if \m{\der{X}\mathcal{M}(X_n)} is invertible. Invertibility  of the matrix  \m{\der{X}\mathcal{M}} is proved for a special case (when only momentum dynamics are considered) in  Appendix \ref{app:nsgrad}. 
    
    Notice that \m{\ocltr{k}} defined in \eqref{eq:mocltr} is not differentiable, but it is Lipschitz continuous. So, we take its generalized gradient~\cite{clarkegg}, defined by 
\[\left(\der{\mcom{k}}\ocltr{k}\right)_{ii}=\begin{cases} -1 \text{\; if\;} \mcom{k}^i \leq \cltb^i, \\ -1 \text{\; if\;} -\mcom{k}^i \geq -\cltb^i,\\\; 0 \phantom{sp} \text{\; otherwise \;}\end{cases} \quad \text{for\;} i = 1,2,3.\]
 The resulting nonlinear algebraic equations \eqref{eq:matching} can be solved using the non-smooth version of Newton's method in \cite{qi}, and this is illustrated with the help of the flowchart shown in Figure \ref{fig:newton}.

\begin{figure}
\centering
\begin{tikzpicture}[scale=0.5,node distance = 1.5cm, auto]
% Define block styles
\tikzstyle{decision} = [diamond, draw, fill=blue!20, 
    text width=5em, text badly centered, node distance=3cm, inner sep=0pt]
\tikzstyle{block} = [rectangle, draw, fill=blue!20, 
    text width=12em, text centered, rounded corners, minimum height=5em]
\tikzstyle{line} = [draw, -latex']
\tikzstyle{cloud} = [draw, ellipse,fill=red!20, node distance=3cm,
    minimum height=2em]    
% Place nodes
    \node [block] (init) {Initialize \m{X} and the Error Tol \m{\delta>0}};
    \node [cloud, above of=init,node distance=1.7cm] (Start) {Start};
    \node [block, below of=init, node distance=2.4cm] (evaluate) {Evaluate the Matching Conditions \m{\mathcal{M}(X)} and its gradient \m{\der{X}\mathcal{M}(X)}};
    \node [block, left of=evaluate, node distance=5.1cm] (update) {update \m{X=X+\Delta X}};
    \node [decision, below of=evaluate,node distance=3cm] (decide) {Is \m{\norm{\mathcal{M}(X)}_{\infty}\leq \delta} ? };
    \node [block, left of= decide, node distance=5.1cm] (solve) {Evaluate \m{\Delta X} by solving the Linear System: \\ \m{\der{X}\mathcal{M}  (X)\Delta X = - \mathcal{M}(X)}};
    
    \node [cloud, below of=decide, node distance=2.3cm] (stop) {Stop};
% Draw edges
    \path [line] (init) -- (evaluate);
    \path [line] (evaluate) -- (decide);
    \path [line] (decide) -- node [near start] {No} (solve);
    \path [line] (solve) -- (update);
    \path [line] (update) |- (evaluate);
    \path [line] (decide) -- node {Yes}(stop);
    \path [line] (Start) -- (init);    
\end{tikzpicture}
\caption{Newton's root finding algorithm} 
\label{fig:newton}
\end{figure}
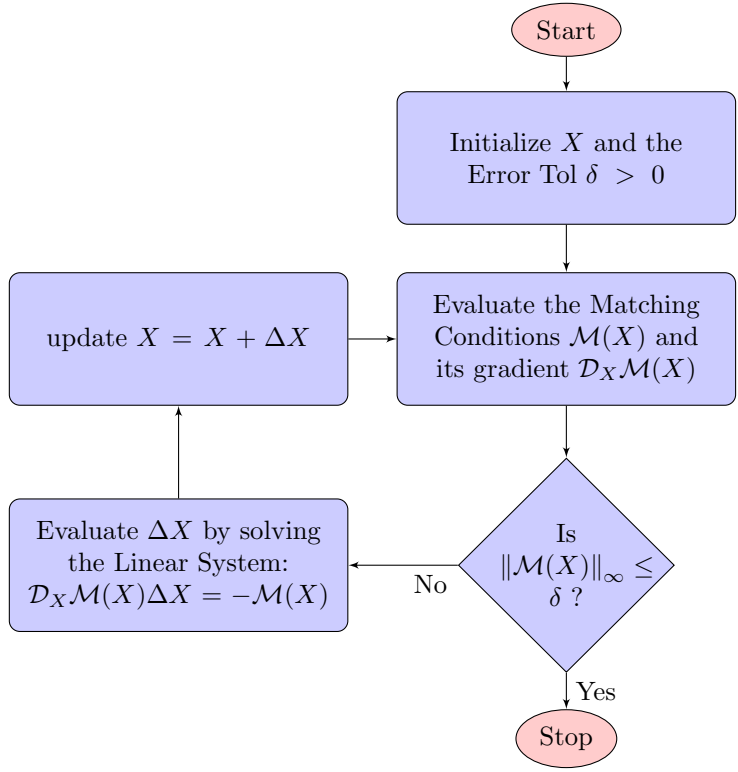

\subsubsection{Multiple shooting method with state constraints}
Here we discuss the general case in which momentum constraints are active. By the complementary slackness conditions \eqref{eq:inq}, it is clear that if the momentum constraints are not active, then \m{\beta_{k}^i = 0} for all \m{i=1,2,3,} and \m{k=1,2,\ldots,N-1.} On the other hand, if the slack variable \m{\mtms{k}^i > 0}, then the momentum variable \m{\mtm{k}^i} corresponding to the slack variable \m{\mtms{k}^i} lies on the boundary (i.e., \m{\mtm{k}^i = \pm \mtmb^i})~\cite{boyd}.
	
    The original Newton's root finding algorithm cannot handle momentum inequality constraints; we propose a modified strategy for handling the momentum inequality constraints. The algorithm has three phases: 
    \begin{enumerate}
	\item \textbf{Projection of the states to the feasible region}: After each Newton's iteration, the momentum variables \m{\{\mtm{k}\}^{N}_{k=0}} are projected onto the feasible domain. 
    \item \textbf{Identification of active constraints:} Active constraints can be identified by computing \m{\beta^i_k} from \eqref{eq:sys}. If \m{\beta^i_k>0}, then the inequality constraint corresponding to the state \m{\mtm{k}^i} is active; else it is inactive.  
    \item \textbf{Enforcing active constraints:} When \m{\beta_k^i > 0}, then the equation with \m{\beta_k^i} in \eqref{eq:sys} is trivially satisfied. So, the matching condition \m{\Xi_k^i} will be replaced with the corresponding active constraint \m{\abs{\mtm{k}^i} = \mtmb^i} for the next iteration. 
	\end{enumerate}
    
     Let 
    \[ X \defas \begin{pmatrix}\mtm{0}^\top & \mcom{0}^\top &\ldots &\mtm{k}^\top &\mcom{k}^\top &\ldots &\mtm{N}^\top &\mcom{N}^\top &\mcor{0}^\top \end{pmatrix}^\top \in \R^{6N+9}\] 
    be a vector in the augmented space; the feasible region \m{\Lambda} is defined as
    \[\Lambda \defas \left\{ X \in \R^{6N+9}\; \left| \; \abs{\mtm{k}^i} \leq \mtmb^i \right.\right\}.\] 
    It is clear that \m{\Lambda } is the intersection of the two half spaces \m{X^{+} \defas \left\{ X \left|\; \mtm{k}^i \leq \mtmb^i \right.\right\}} and \m{X^{-} \defas \left\{ X \left|\; \mtm{k}^i \geq -\mtmb^i \right.\right\}}, hence a convex set. Let \m{\mathcal{P}_{\Lambda} : \R^{6N +9} \rightarrow \Lambda} be the projection onto the feasible region \m{\Lambda},  and \m{\tilde{\mathcal{M}}} represents the active constraints, which is a collection of equality state constraints, active inequality state constraints and co-state constraints corresponding to inactive inequality state constraints. Our modified multiple shooting method is described in Algorithm \ref{mmsm}, and is illustrated with the help of the flowchart given in Figure \ref{fig:mnewton}.
\begin{spacing}{1}
\begin{algorithm}
\caption{Modified multiple shooting method}\label{mmsm}
\begin{algorithmic}[1]
\Procedure{Root of the function \m{\tilde{\mathcal{M}}} such that \m{X \in \Lambda}}{}
\State Choose appropriate initial guess \m{X} and the error tolerance bound \m{\delta>0}.
\State \textbf{Projection:}
\State Compute \m{X \mapsto \mathcal{P}_{\Lambda}(X) \defas \begin{pmatrix}\pmtm{0}^\top & \mcom{0}^\top &\ldots &\pmtm{k}^\top &\mcom{k}^\top &\ldots &\pmtm{N}^\top &\mcom{N}^\top &\mcor{0}^\top \end{pmatrix}^\top } \phantom{spac} where \m{\pmtm{k}^i = \sgn(\mtm{k}^i)\min\left\{\mtmb^i, \abs{\mtm{k}^i}\right\}}.
\State Using \eqref{eq:sys}, update co-states 
\begin{align*}
\pmcom{k-1}^i = \begin{cases}-\ocltr{k}^i \quad \text{\;if\;} \abs{\ocltr{k}^i}<\cltb^i,\\ \mcom{k-1}^i \quad \text{elsewhere},\end{cases}\quad  \text{where}\;
\ocltr{k} =\frac{\left(\pmtm{k}-\rot{k-1}^\top\pmtm{k-1}\right)}{h}.
\end{align*}
\If {\m{\pmtm{k}^i = \mtmb^i}}
\State Compute the slack variable using \eqref{eq:sys} as
\begin{align*} 
\phantom{space}\beta_{k}^i = \frac{\left(Z(\pmtm{k})\Xi_{k-1}(\pmtm{k},\pmcom{k}, \pmcom{k-1}, \mcor{0})\right)^i}{\pmtm{k}^i}, \text{\;\;where\;\;} Z(\mtm{k}) =\left(\frac{\rot{k}}{h}- \mathcal{N}_{k} \widehat{\rot{k}^\top\mtm{k}}\right).
\end{align*} 
\EndIf
\State Define active constraints 
\begin{align*}
\tilde{\mathcal{M}}(X) & \defas \begin{pmatrix} \Sigma_1 &\tilde{\Xi}_1 & \ldots & \Sigma_k & \tilde{\Xi}_k & \ldots & \Sigma_{N} & \tilde{\Xi}_{N} & \Mtmcnst^\top & \Orntcnst^\top \end{pmatrix}^\top, \\
\text{where} & \quad \\
\tilde{\Xi}_k^i &\defas \begin{cases} \abs{\mtm{k}^i} - \mtmb^i \quad \text{if\;} \mtms{k}^i > 0, \\ \Xi_k^i \qquad \qquad \text{otherwise}.\nsp\end{cases}
\end{align*}
\State Evaluate the matching conditions \m{\tilde{\mathcal{M}}(\tilde{X})} and its gradient \m{\der{X}\tilde{\mathcal{M}}(\tilde{X})}, \phantom{spaci}where \m{\tilde{X} \defas \begin{pmatrix}\pmtm{0}^\top & \pmcom{0}^\top &\ldots &\pmtm{k}^\top &\pmcom{k}^\top &\ldots &\pmtm{N}^\top &\pmcom{N}^\top &\mcor{0}^\top \end{pmatrix}^\top}.
\If {\m{\norm{\tilde{\mathcal{M}}(\tilde{X})}_{\infty} < \delta}} Stop.  
\EndIf
\State Calculate \m{\Delta \tilde{X}} by solving the linear system \m{\der{X}\tilde{\mathcal{M}}(\tilde{X})\Delta \tilde{X} = - \tilde{\mathcal{M}}(\tilde{X})}. 
\State Update the value of \m{X=\tilde{X}+\Delta \tilde{X}.} 
\State \textbf{goto} \textbf{Projection}.
\EndProcedure
\end{algorithmic}
\end{algorithm}
\end{spacing}

\begin{figure}
\centering
\begin{tikzpicture}[scale=0.5,node distance = 2cm, auto]
% Define block styles
\tikzstyle{decision} = [diamond, draw, fill=blue!20, 
    text width=5em, text badly centered, node distance=3cm, inner sep=0pt]
\tikzstyle{sblock} = [rectangle, draw, fill=blue!20, 
    text width=12em, text centered, rounded corners, minimum height=3em]
\tikzstyle{block} = [rectangle, draw, fill=blue!20, 
    text width=17em, text centered, rounded corners, minimum height=3em]
\tikzstyle{line} = [draw, -latex']
\tikzstyle{cloud} = [draw, ellipse,fill=red!20, node distance=3cm,
    minimum height=2em]    
% Place nodes
    \node [block] (init) {Initialize \m{X} and the Error Tol \m{\delta>0}};
    \node [cloud, above of=init,node distance=1.3cm] (Start) {Start};
    \node [block, below of=init, node distance=1.5cm] (project){Evaluate \m{\mathcal{P}_{\Lambda}(X), \pmcom{k}, \mtms{k}}};
    \node [block, below of=project, node distance=1.5cm](updtcst){Define active constraints \m{\tilde{\mathcal{M}}} and \m{\tilde{X}}};
    \node [block, below of=updtcst, node distance=1.7cm] (evaluate) {Evaluate the matching conditions \m{\tilde{\mathcal{M}}(\tilde{X})} and its gradient \m{\der{X}\tilde{\mathcal{M}}(\tilde{X})}};
    \node [sblock, left of=project, node distance=6.2cm] (update) {update \m{X=\tilde{X}+\Delta \tilde{X}}};
    \node [sblock, left of= evaluate, node distance=6.2cm] (solve) {Evaluate \m{\Delta \tilde{X}} by solving the Linear System: \\ \m{\der{X}\tilde{\mathcal{M}} (\tilde{X})\Delta \tilde{X} = - \tilde{\mathcal{M}}(\tilde{X})}};
    \node [decision, below of=evaluate,node distance=2.9cm] (decide) {Is \m{\norm{\tilde{\mathcal{M}}(\tilde{X})}_{\infty}\leq \delta} ? };
    \node [cloud, below of=decide, node distance=2.5cm] (stop) {Stop};
% Draw edges
    \path [line] (init) -- (project);
    \path [line] (project) -- (updtcst);
    \path [line] (updtcst) -- (evaluate);   
    \path [line] (evaluate) -- (decide);  
    \path [line] (decide) -| node [near start] {No} (solve);
    \path [line] (solve) -- (update);
    \path [line] (update) |- (project);
    \path [line] (decide) -- node {Yes}(stop);
    \path [line] (Start) -- (init);    
\end{tikzpicture}
\caption{Modified multiple shooting algorithm for problems with state constraints}
\label{fig:mnewton}
\end{figure}

    \section{Numerical Experiments}\label{s:nsr}
    The following data have been considered for the simulations: Principal Moment of inertia of the satellite \m{\left(\ii{x},\ii{y},\ii{z}\right)}=\unit{\left(800,1200,1000\right)}{\kilogram\squaren\meter}, Sampling time ($T$) = \unit{0.1}{\second}, Range of angles ($\theta$) can vary between [\unit{10}{\degree}, \unit{90}{\degree}] about any axis, Maximum torque or control bound ($\cltb$) = \unit{\left(20,20,20\right)}{\newton\meter}, maximum momentum ($\mtmb$) = \unit{\left(70,70,70\right)}{\newton\meter\second}, time duration ($t_{\max}$) can range from \unit{5}{\second} to \unit{30}{\second}. 
    
    	An open loop optimal control profile is obtained for a maneouvre of reorienting the satellite for \m{\unit{90}{\degree}} about  axis \m{\left(\frac{1}{\sqrt{3}},\frac{1}{\sqrt{3}},\frac{1}{\sqrt{3}}\right)} from initial momentum $\mtm{i}$ = \unit{\left(30,-10,10\right)}{\newton\meter\second} to desired momentum  $\mtm{f}$ = \unit{\left(0,0,0\right)}{\newton\meter\second} in  \unit{19}{\second}. The optimal control profile along with the momentum and co-state vectors is shown in Figure~\ref{fig:CRot90_111_19}.
    
\begin{figure}[H]
\centering  
\includegraphics[width=0.8\textwidth]{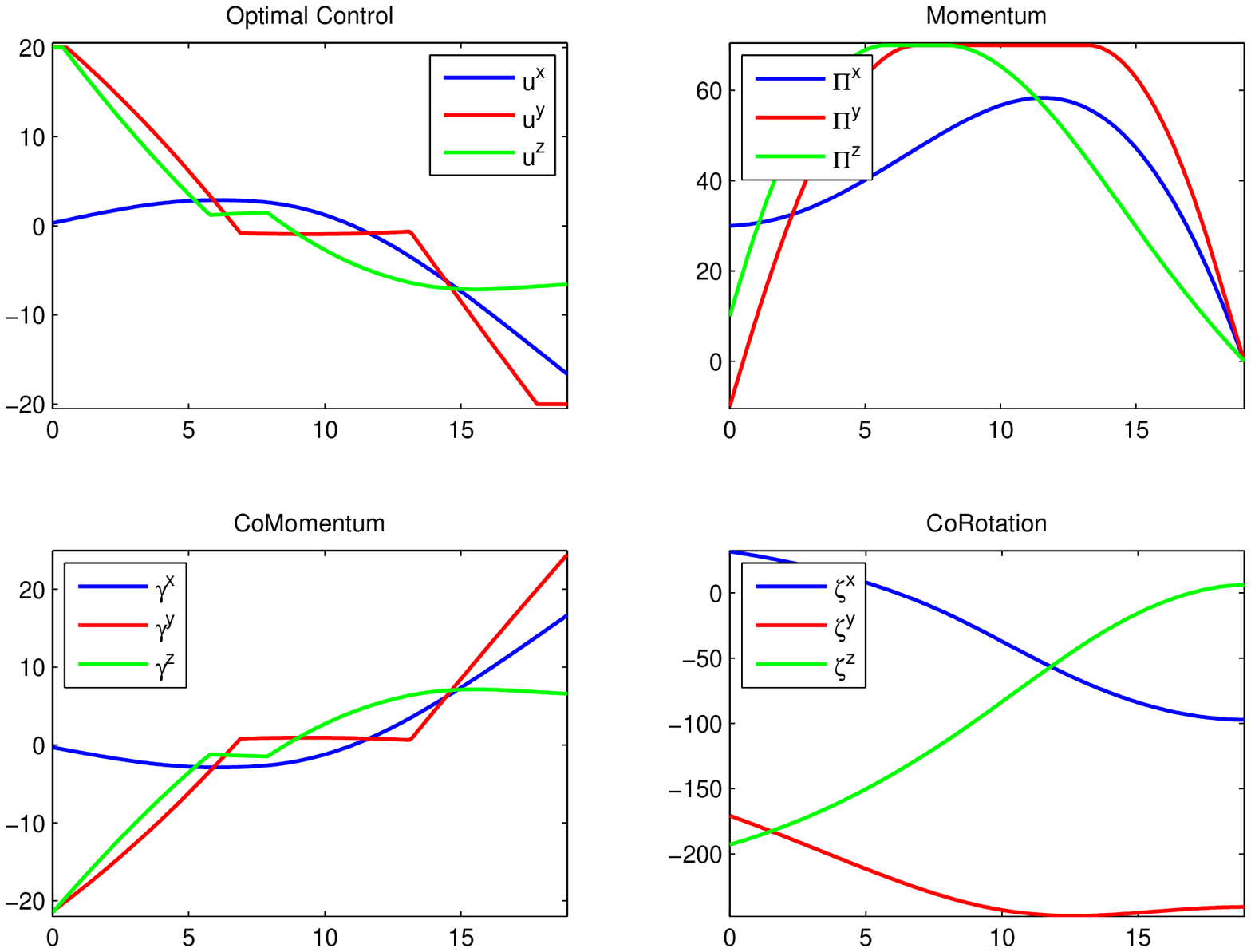}
\caption{Rotation of \m{\unit{90}{\degree}} about \m{\left(\frac{1}{\sqrt{3}},\frac{1}{\sqrt{3}},\frac{1}{\sqrt{3}}\right)} axis with initial momentum $\mtm{i}$ = \unit{\left(30,-10,10\right)}{\newton\meter\second}, desired momentum  $\mtm{f}$ = \unit{\left(0,0,0\right)}{\newton\meter\second}, time duration = \unit{19}{\second}.}
\label{fig:CRot90_111_19}
\end{figure}

It is clear from the Figure~\ref{fig:CRot90_111_19} that the control saturates when the absolute values of the co-state vectors corresponding to the momentum are more than the control bounds. If the momentum saturates, then the control along that axis goes to zero.  
    
    \section{Conclusion and Future Direction}\label{s:cfd}
A tractable solution to a class of optimal control of spacecraft attitude maneuvers under control and momentum constraints has been presented. We used discrete mechanics to discretize the continuous time model, which has many advantages over the conventional discretization schemes like Euler's steps. In particular, the model so obtained preserves conserved quantities of the body like momentum and energy. It is efficient in terms of numerics because the model reduces to the momentum dynamics only. A new multiple shooting algorithm has been proposed to solve the system of difference equations obtained as first order necessary conditions using variational analysis. This algorithm can be used to solve optimal control problems with state inequality constraints. In future, convergence results for the modified multiple shooting methods will be explored. The proposed multiple shooting method uses Newton's root finding algorithm at each step for finding the roots of the matching conditions. Newtons's method has a quadratic rate of convergence in a neighborhood of the true solution, and therefore it is worth exploring the convergence rates of our algorithm in a neighborhood of the matching conditions. It is important to note that our algorithm can handle box constraints only; an extension to more general constraints is under development.    
        
   %\clearpage
    \appendix

%(See page~\pageref{lemma:invFkJd})
\section{Proof of Lemma~\ref{lemma:invFkJd} } \label{app:invFkJd}
\invFkJd*
\begin{proof}
If \m{\trace\left(\rot{k}\mint\right) = 0} then the matrix \m{\mtrt{k}} is invertible. If \m{\trace\left(\rot{k}\mint\right) \neq 0} then \m{\mtrt{k}} is invertible if and only if the matrix \m{\left(\I - \frac{\rot{k} \mint}{\trace\left(\rot{k}\mint\right)}\right)} is invertible. Using Banach lemma ~\cite[p.\ 193]{lax}, the matrix \m{\left(\I - \frac{\rot{k} \mint}{\trace\left(\rot{k}\mint\right)}\right)} is invertible if \m{\norm{\frac{\rot{k} \mint}{\trace\left(\rot{k}\mint\right)}} < 1.}
\[ \norm{\frac{\rot{k} \mint}{\trace\left(\rot{k}\mint\right)}} = \frac{\norm{\rot{k} \mint}}{\abs{\trace\left(\rot{k}\mint\right)}} \leq \frac{\norm{\mint}}{\abs{\trace\left(\rot{k}\mint\right)}}\]
Then \m{\norm{\frac{\rot{k} \mint}{\trace\left(\rot{k}\mint\right)}} < 1 \quad \text{holds if} \quad \norm{\mint} < \abs{\trace\left(\rot{k}\mint\right)}}.
Now we show that if \m{\cos \left(\frac{\norm{\xi_k}}{2}\right) < \sqrt{\frac{2 d_3 + d_2 -d_1}{2(d_3 + d_2)}}} then condition \m{\norm{\mint} < \abs{\trace\left(\rot{k}\mint\right)}} is true.
Let us choose the quaternions \m{\left(q_0, q_1, q_2, q_3\right) \in \R^4} as a parametrization for the rotation matrix \m{\rot{k}} where \m{q_0 = \cos\left(\frac{\norm{\xi_k}}{2}\right)} and \m{\left(q_1, q_2, q_3 \right)^\top = \xi_k^e \sin\left(\frac{\norm{\xi_k}}{2}\right)}, the vector \m{\xi_k^e } represent the unit vector corresponding to vector \m{\xi_k.} Then
\begin{align} \nonumber
\abs{\trace\left(\rot{k}\mint\right)} &= \abs{d_1 (q_0^2 + q_1^2 -q_2^2 -q_3^2) + d_2 (q_0^2 - q_1^2 + q_2^2 -q_3^2) + d_3 (q_0^2 - q_1^2 -q_2^2 + q_3^2) } \\\nonumber
&\geq d_1 (q_0^2 + q_1^2 -q_2^2 -q_3^2) + d_2 (q_0^2 - q_1^2 + q_2^2 -q_3^2) + d_3 (q_0^2 - q_1^2 -q_2^2 + q_3^2) \\\nonumber
& = d_1 + d_2 + d_3 -2 \left\{ (d_1 + d_2) q_1^2 + (d_2 + d_3) q_2^2 + (d_3 + d_4) q_3^2 \right\} \\
& \geq d_1 + d_2 + d_3 -2 \tilde{f}, \label{eq:trFkJd}
\end{align}
where 
\begin{align}\label{eq:optmz}
\tilde{f} = &\maximize_{x,y,z} \quad (d_1 + d_2) x + (d_2 + d_3) y + (d_3 + d_4) z \\ \nonumber
 \text{subject to}&  \\ \nonumber
 &  x+y+z = 1-q_0^2, \\ \nonumber
 &  x \geq 0, y \geq 0, z \geq 0.
\end{align}
The optimization problem defined in \eqref{eq:optmz} is a linear programming problem and hence the optimum value will be attained on the vertices of the feasible region. Without loss of generality, assume that \m{0 < d_1 \leq d_2 \leq d_3.} Then \m{\tilde{f} = (d_3 + d_2)(1-q_0^2).}
Given that \m{\cos \left(\frac{\norm{\xi_k}}{2}\right) < \sqrt{\frac{2 d_3 + d_2 -d_1}{2(d_3 + d_2)}}} and \m{q_0 = \cos\left(\frac{\norm{\xi_k}}{2}\right)}, one can conclude that \m{q_0^2 < \frac{2 d_3 + d_2 -d_1}{2(d_3 + d_2)}.} Substituting the value of \m{\tilde{f}} to \eqref{eq:trFkJd} gives
\[\abs{\trace\left(\rot{k}\mint\right)} \geq d_1 + d_2 + d_3 -2 (d_3 + d_2)(1-q_0^2) > d_3 = \norm{\mint}. \]
\end{proof}
%(See page~\pageref{claim:fskew})
\section{Proof of Claim~\ref{claim:fskew} } \label{app:fskew}
\fskew*
\begin{proof}
This claim can be proved using the Rodrigues's formula. Let \m{A \defas\ornt{k}^\top \ornt{k+1} \in \SO{3}} and \m{B \defas \rot{k} \in \SO{3}} then there exist vectors \m{a, b \in \R^3} such that 
\begin{align}\label{eq:exp}
A=\e^{(\norm{a} \widehat{a}_e)} \quad \text{and} \quad B=\e^{(\norm{b} \widehat{b}_e)}
\end{align}
where \m{x_e} is a unit vector corresponding to the vector \m{x \in \R^3}. Using the Rodrigues's formula ~\cite{blochsymmetric}
\[  A=\e^{(\norm{a} \widehat{a}_e)} \defas \I + \sin(\norm{a})\widehat{a}_e + \widehat{a}_e \widehat{a}_e^\top (\cos(\norm{a})-1) \]
we obtain, the skew-symmetric parts of \m{A} and \m{B} as
\begin{align}\label{eq:skewab}
S(A) \defas \frac{A-A^\top}{2} = \sin(\norm{a})\widehat{a}_e \quad \text{and} \quad S(B)\defas\frac{B-B^\top}{2} = \sin(\norm{b})\widehat{b}_e .
\end{align}
Given that \m{S(A) = S(B)} and since we know that \m{\so{3} \iso \R^3}, from \eqref{eq:skewab} we have
\begin{align}\label{eq:veceq}
 \sin(\norm{a})a_e = \sin(\norm{b})b_e 
\end{align}
\underline{\textbf{case 1:}}\\
If \m{\norm{b} = 0} then from \eqref{eq:veceq} \m{\sin(\norm{a})a_e = 0 } iff \m{a=0} which means \m{a=b} because \m{\norm{a} \leq \frac{\pi}{2}.} Using the exponential map \eqref{eq:exp} one can conclude that \m{A=B.}\\
\underline{\textbf{case 2:}}\\
If \m{\norm{b} \neq 0} then  from \eqref{eq:veceq} 
\[ \frac{\sin(\norm{a})}{\sin(\norm{b})}a_e = b_e\]
which is true iff 
\begin{align}\label{eq:unit}
& b_e=a_e \quad \text{and} \quad \\ 
& \frac{\sin(\norm{a})}{\sin(\norm{b})} = \pm 1.\label{eq:sin}
\end{align}
The map \m{\sin : ] -\frac{\pi}{2}, \frac{\pi}{2} [ \rightarrow ]-1,1[} is bijective. Given that \m{\norm{b} < \frac{\pi}{2}}, we conclude from \eqref{eq:sin} that \m{\sin(\norm{a}) = \sin(\norm{b})}. Hence 
\begin{align}\label{eq:norm}
\norm{b}= \norm{a}
\end{align}
From \eqref{eq:unit} and \eqref{eq:norm} it is clear that \m{a = b}. Using the exponential map \eqref{eq:exp}, one can conclude that \m{A=B.}
\end{proof}
%(See page~\pageref{lemma:invl})

% (See page~\pageref{thm:nsgrad})
\section{Proof of Theorem~\ref{thm:nsgrad}}\label{app:nsgrad}
Let us discuss a case in which, only momentum dynamics is considered . In this case, matching conditions \eqref{eq:matching} modifies to
\begin{align}\label{eq:mmatching}
\mathcal{C}(Y) \defas 
\begin{pmatrix}
\Sigma_1 & \tilde{\Xi}_1 & \ldots& \Sigma_k& \tilde{\Xi}_k&  \ldots& \Sigma_{N}& \tilde{\Xi}_{N}&
\Mtmcnst^\top \end{pmatrix}^\top = 0,
\end{align}
where 
\begin{align*}
&Y \defas \begin{pmatrix}\mtm{0}^\top & \mcom{0}^\top &\ldots \mtm{k}^\top &\mcom{k}^\top &\ldots &\mtm{N}^\top &\mcom{N}^\top \end{pmatrix}^\top,\;
\Sigma_{k+1} \defas \left(\mtm{k+1}-\rot{k}^\top\mtm{k} - h \ocltr{k}\right)^\top, \\
& \tilde{\Xi}_{k+1} \defas \left(-\mcom{k} + \left(\rot{k+1}-h\mathcal{N}_{k+1} \widehat{\rot{k+1}^\top \mtm{k+1}}\right) \mcom{k+1} \right)^\top , \;
\Mtmcnst \defas \begin{pmatrix} \mtm{N} - \mtm{f} \\ \mtm{0} - \mtm{i} \end{pmatrix}. 
\end{align*}
We define the gradient of the matching conditions \eqref{eq:mmatching} as
\begin{align}
\der{Y}\mathcal{C}(Y) \defas 
\begin{pmatrix} -\mathcal{J}_1 & \mathcal{K}_1 & 0 & \ldots & 0& 0 & 0\\
         0 &-\mathcal{J}_2 & \mathcal{K}_2 & \ldots & 0& 0 & 0\\        
        \vdots&\vdots&\vdots&\vdots&\vdots&\vdots&\vdots\\
        0 & 0 & 0 & \ldots & -\mathcal{J}_{N-1} & \mathcal{K}_{N-1} & 0 \\
        0 & 0 & 0 & \ldots & 0 & -\mathcal{J}_{N} & \mathcal{K}_{N} \\
        \mathcal{B}^i & 0 & 0 & \ldots & 0 & 0 & \mathcal{B}^f       
\end{pmatrix},
\end{align}
where
\[\mathcal{J}_k \defas \begin{pmatrix}-\der{\mtm{k-1}}\Sigma_k^\top & -\der{\mcom{k-1}}\Sigma_k^\top \\ 0 & \I \end{pmatrix}, \quad
\mathcal{K}_k \defas \begin{pmatrix} \I & 0 \\\der{\mtm{k}}\tilde{\Xi}_k^\top &\der{\mcom{k}}\tilde{\Xi}_k^\top \end{pmatrix},\]
\[\mathfrak{B}^i \defas \der{\mtm{0}}\Mtmcnst =  \begin{pmatrix} 0 \!&\! 0 \\ \I \!&\! 0 \end{pmatrix}, \quad
\mathfrak{B}^f \defas \der{\mtm{N}}\Mtmcnst = \begin{pmatrix} \I \!&\! 0 \\ 0 \!&\! 0 \end{pmatrix}.\]
\begin{restatable}{theorem}{nsgrad}
\label{thm:nsgrad}
The Jacobian matrix \m{\der{Y}\mathcal{C}} is non-singular if there exist a time instant \m{k} such that the control \m{\ocltr{k}} is not saturated. 
\end{restatable}
\begin{proof}
The matrix \m{\der{Y}\mathcal{C}} is similar to the upper triangular matrix \m{\der{Y}\widetilde{\mathcal{C}}} which is obtained using the elementary row transformations. 
\begin{align}
\der{Y}\widetilde{\mathcal{C}}(Y) \defas 
\begin{pmatrix} -\mathcal{J}_1 & \mathcal{K}_1 & 0 & \ldots & 0& 0 & 0\\
         0 &-\mathcal{J}_2 & \mathcal{K}_2 & \ldots & 0& 0 & 0\\        
        \vdots&\vdots&\vdots&\vdots&\vdots&\vdots&\vdots\\
        0 & 0 & 0 & \ldots & -\mathcal{J}_{N-1} & \mathcal{K}_{N-1} & 0 \\
        0 & 0 & 0 & \ldots & 0 & -\mathcal{J}_{N} & \mathcal{K}_{N} \\
        0 & 0 & 0 & \ldots & 0 & 0 & \mathcal{X}       
\end{pmatrix},
\end{align}
Note that the matrix \m{\der{Y}\widetilde{\mathcal{C}}(Y)} is invertible if and only if \m{\mathcal{J}_k} is invertible for \m{k = 1,2, \ldots, N}, and \m{\mathcal{X}} is invertible.
First we prove that \m{\mathcal{J}_k} is invertible,
\begin{align}
\mathcal{J}_k &
=\begin{pmatrix}-\der{\mtm{k-1}}\Sigma_k^\top & -\der{\mcom{k-1}}\Sigma_k^\top \\ 0 & \I \end{pmatrix} 
= \begin{pmatrix}\left(\rot{k-1}-h\mathcal{N}_{k-1} \widehat{\rot{k-1}^\top\mtm{k-1}}\right)^\top & h\der{\mcom{k-1}}\ocltr{k-1}\\ 0 & \I \end{pmatrix}. \label{eq:gi}
\end{align}
\begin{restatable}{lemma}{invl}
\label{lemma:invl}
The matrix \m{\left(\rot{k}-h\mathcal{N}_k \widehat{\rot{k}^\top\mtm{k}}\right)} is invertible if and only if the matrix \m{\mtrt{k}} is invertible.
\end{restatable}
A proof of Lemma \ref{lemma:invl} is provided in Appendix \ref{app:invl}.
By Lemma~\ref{lemma:invl} and Lemma~\ref{lemma:invFkJd}, we know that the matrix \m{\left(\rot{k}-h\mathcal{N}_{k} \widehat{\rot{k}^\top\mtm{k}}\right)} is invertible. Hence, from \eqref{eq:gi} it is easy to conclude that the matrix \m{\mathcal{J}_k} is invertible.
Let us define the inverse of the matrix \m{\mathcal{J}_k} as 
\[ \mathcal{J}_k^{-1} \defas \begin{pmatrix} a_{k-1}^{-1} & - a_{k-1}^{-1} b_{k-1}  \\ 0 & \I \end{pmatrix}, \quad \mathcal{K}_k \defas \begin{pmatrix} \I & 0  \\ c_k & a_{k}^\top \end{pmatrix},\]
where \[a_k = \left(\rot{k}-h\mathcal{N}_{k} \widehat{\rot{k}^\top\mtm{k}}\right)^\top, \quad b_k = h\der{\mcom{k}}\ocltr{k}, \quad c_k = \der{\mtm{k}}\tilde{\Xi}_k.\]
Taking the first order approximation in \m{h}, the matrix \m{\mathcal{X}} can be approximated as
\begin{align}
\mathcal{X} &= \mathcal{B}^f + \mathcal{B}^i\mathcal{J}_{1}^{-1} \mathcal{K}_{1}\mathcal{J}_{2}^{-1} \mathcal{K}_{2} \ldots \mathcal{J}_{N}^{-1} \mathcal{K}_{N} \nonumber \\
& \approx \begin{pmatrix} \I & 0 \\ * & s \end{pmatrix}, \label{eq:z}
\end{align}
where \m{s = - r_0^{-1}\left(\sum_{i=0}^{N-1} r_{i+1} b_i r_{i+1}^\top \right)} such that \m{r_{i} = a_N a_{N-1}\ldots a_i.}
Now it is clear from \eqref{eq:z} that the matrix \m{\mathcal{X}} is invertible if the matrix \m{s} is invertible.
As the matrix \m{a_i} is invertible, so the matrix \m{r_{i}} is invertible for all {i = 0,1,\ldots,N}. So, it is enough to prove that the matrix \m{\sum_{i=0}^{N-1} r_{i+1} b_i r_{i+1}^\top } is invertible for \m{s} being invertible. Note that the matrix \m{b_i} is a negative semi-definite matrix for \m{i=1,2,\ldots,N-1.} Assuming that there exist a time instant \m{k} such that the control \m{\ocltr{k}} is not saturated means that \m{b_k} is negative definite. So, it is concluded that the matrix \m{\sum_{i=0}^{N-1} r_{i+1} b_i r_{i+1}^\top } is negative definite and hence invertible.  
\end{proof}
\section{Proof of Lemma~\ref{lemma:invl} }\label{app:invl}
\invl*
\begin{proof}
We know that \m{\widehat{\rot{k}^\top h\mtm{k}} = \rot{k}^\top \widehat{h\mtm{k}} \rot{k}}. Then using the implicit equation \m{\widehat{h\mtm{k}}= \rot{k}\mint - \mint \rot{k}^T,} the matrix \m{\left(\rot{k}-h\mathcal{N}_k \widehat{\rot{k}^\top\mtm{k}}\right)} can be rewritten as 
\begin{align}\nonumber
\left(\rot{k}-h\mathcal{N}_k \widehat{\rot{k}^\top\mtm{k}}\right) &= \left(\rot{k}- \mathcal{N}_k \rot{k}^\top \widehat{h\mtm{k}}\rot{k}\right)\\\nonumber
& = \mathcal{N}_k \rot{k}^\top  \left(\rot{k} \left(\mathcal{N}_k\right)^{-1} - \widehat{h\mtm{k}} \right) \rot{k} \\\nonumber
& = \mathcal{N}_k \rot{k}^\top  \left(\trace\left(\mint \rot{k}^\top\right) \I - \mint \rot{k}^\top  - \widehat{h\mtm{k}} \right) \rot{k} \\ \nonumber
& = \mathcal{N}_k \rot{k}^\top  \left(\trace\left(\rot{k} \mint\right) \I - \rot{k} \mint \right) \rot{k} \\
&= \left( \trace\left(\mint \rot{k}^\top\right) \I - \mint \rot{k}^\top\right)^{-1} \mtrt{k} \rot{k}. \label{eq:invHam}
\end{align}
By \eqref{eq:invHam} it is obvious to conclude that the matrix \m{\mtrt{k}} is invertible if and only if the matrix \m{\left(\rot{k}-h\mathcal{N}_k \widehat{\rot{k}^\top\mtm{k}}\right)} is invertible.
\end{proof}

\end{document}